\documentclass[
a4paper,UKenglish, cleveref, autoref, thm-restate]{lipics-v2021}
\title{Dynamic size counting in population protocols} %TODO Please add

\author{David Doty}{University of California, Davis \and \url{https://web.cs.ucdavis.edu/~doty/}}
{doty@ucdavis.edu}{https://orcid.org/0000-0002-3922-172X}{}%TODO mandatory, please use full name; only 1 author per \author macro; first two parameters are mandatory, other parameters can be empty. Please provide at least the name of the affiliation and the country. The full address is optional

\author{Mahsa Eftekhari}{University of California, Davis \and \url{https://eftekhari.cs.ucdavis.edu/}}
{mhseftekhari@ucdavis.edu}{https://orcid.org/0000-0001-5680-2086}{}

\authorrunning{Doty, Eftekhari} %TODO mandatory. First: Use abbreviated first/middle names. Second (only in severe cases): Use first author plus 'et al.'

\Copyright{David Doty and Mahsa Eftekhari} %TODO mandatory, please use full first names. LIPIcs license is "CC-BY";  http://creativecommons.org/licenses/by/3.0/

 %TODO mandatory; please add comma-separated list of keywords

\category{} %optional, e.g. invited paper

\relatedversion{} %optional, e.g. full version hosted on arXiv, HAL, or other respository/website
\funding{Supported by NSF award 1900931 and CAREER award 1844976.}
\nolinenumbers %uncomment to disable line numbering

\hideLIPIcs  %uncomment to remove references to LIPIcs series (logo, DOI, ...), e.g. when preparing a pre-final version to be uploaded to arXiv or another public repository

\EventEditors{}
\EventNoEds{2}
\EventLongTitle{1st Symposium on Algorithmic Foundations of Dynamic Networks (SAND 2022)}
\EventShortTitle{2022}
\EventAcronym{SAND}
\EventYear{2022}
\EventDate{March 28--30, 2022}
\EventLocation{}
\EventLogo{}
\SeriesVolume{}
\ArticleNo{}
\usepackage{cite}

\usepackage[appendix=inline]{apxproof}

\usepackage{physics}
\newcommand{\timer}{\ensuremath{\mathtt{timer}}}
\newcommand{\phase}{\tt{phase}}
\newcommand{\group}{\ensuremath{\mathtt{group}}}
\newcommand{\RG}{\ensuremath{\mathtt{GRV}}}
\newcommand{\missingV}{\ensuremath{\mathtt{FMV}}}
\newcommand{\logMissingV}{\ensuremath{\mathtt{LFMV}}}
\newcommand{\signalArray}{\ensuremath{\mathtt{signals}}}
\newcommand{\est}{\ensuremath{\mathtt{estimate}}}

\newcommand{\av}{\ensuremath{\mathrm{v}}}
\newcommand{\au}{\ensuremath{\mathrm{u}}}
\newcommand{\agent}{\ensuremath{\mathrm{agent}}}
\newcommand{\nPhase}{\ensuremath{\mathsf{Normal Phase}}}
\newcommand{\wPhase}{\ensuremath{\mathsf{Waiting Phase}}}
\newcommand{\cPhase}{\ensuremath{\mathsf{Updating Phase}}}
\newcommand{\updateGroup}{\ensuremath{\mathrm{Update Group}}}
\newcommand{\dcountingA}{{\mathrm{Dynamic Counting}}}
\newcommand{\updateMissingV}{\ensuremath{\mathrm{Update MV}}}
\newcommand{\checkSizeChange}{\ensuremath{\mathrm{Size Checker}}}
\newcommand{\timerRoutine}{\ensuremath{\mathrm{Timer Routine}}}
\newcommand{\EpidemicMaxRG}{\ensuremath{\mathrm{Propagate Max Est}}}
\newcommand{\resetCountDown}{\ensuremath{\mathrm{Signal Propagation}}}

\usepackage[disable,textsize=scriptsize]{todonotes} \newcommand{\todoi}[1]{\todo[inline]{#1}}
\usepackage{algorithm}

\usepackage{caption}
\usepackage[noend]{algpseudocode}
\algnewcommand{\LeftComment}[1]{\Statex \(\triangleright\) #1}

\newcommand{\B}{\mathsf{B}}

\usepackage{mathtools}
\usepackage{svg}

\DeclarePairedDelimiter\ceil{\lceil}{\rceil}
\DeclarePairedDelimiter\floor{\lfloor}{\rfloor}

\usepackage{amsfonts}
\let\vec\mathbf

\newcommand{\N}{\mathbb{N}}

\newcommand{\vc}{\vec{c}}

\newcommand{\calP}{\mathcal{P}}
\newcommand{\calA}{\mathcal{A}}
\newcommand{\calE}{\mathcal{E}}
\newcommand{\calD}{\mathcal{D}}
\newcommand{\rto}{\rightarrow}
\renewcommand{\Pr}[1]{\mathrm{Pr}\left[#1\right]}

\renewcommand{\exp}[1]{\mathrm{exp}\left(#1\right)}

\usepackage[normalem]{ulem}
\newtheoremrep{theorem}{Theorem}[section]
\newtheoremrep{lemma}[theorem]{Lemma}
\newtheoremrep{corollary}[theorem]{Corollary}
\newtheoremrep{observation}[theorem]{Observation}

\crefname{algorithm}{Protocol}{Protocols} 
\Crefname{algorithm}{Protocol}{Protocols} 

\makeatletter

\begin{document}

\maketitle

\noindent

\begin{abstract}
The population protocol model describes a network of anonymous agents that interact asynchronously in pairs chosen at random.
Each agent starts in the same initial state $s$.
We introduce the \emph{dynamic size counting} problem: 
approximately counting the number of agents in the presence of an adversary who at any time can remove 
any number of 
agents or add 
any number of new 
agents in state $s$.
A valid solution requires that after each addition/removal event, 
resulting in population size $n$, 
with high probability each agent ``quickly'' computes the same constant-factor estimate of the value $\log_2 n$
(how quickly is called the \emph{convergence} time),
which remains the output of every agent for as long as possible
(the \emph{holding} time).
% It is straightforward to show that 
Since the adversary can remove agents,
the holding time is necessarily finite:
even after the adversary stops altering the population,
it is impossible to \emph{stabilize} to an output that never again changes.

We first show that a protocol solves the dynamic size counting problem 
if and only if it solves the \emph{loosely-stabilizing counting} problem:
that of estimating $\log n$ in a \emph{fixed-size} population, 
but where the adversary can initialize each agent in an arbitrary state,
with the same convergence time and holding time.
We then show a protocol solving the loosely-stabilizing counting problem with the following guarantees:
if the population size is $n$,  $M$ is the largest initial estimate of $\log n$,
and $s$ is the maximum integer initially stored in any field of the agents' memory,
we have expected convergence time $O(\log n + \log M)$,
expected polynomial holding time,
and expected memory usage of
$O(\log^2 (s) + (\log \log n)^2)$ bits. \todo{ME: $O((\log n)^{\log n})$ states. How did we get $(\log \log n)^2$ bits? this seems like the result adding optimization, right?}
Interpreted as a dynamic size counting protocol,
when changing from population size $n_\mathrm{prev}$ to $n_\mathrm{next}$,
the convergence time is $O(\log n_\mathrm{next} + \log \log n_\mathrm{prev})$.
\end{abstract}

\section{Introduction}
\label{sec:intro}
A population protocol~\cite{AADFP06} is a network of $n$ anonymous and identical \emph{agents} with finite memory called the \emph{state}.
% that communicate through pairwise random interactions. 
% At every interaction, 
A scheduler repeatedly selects a pair of agents independently and uniformly at random to interact.
Each agent sees the entire state of the other agent in the interaction and 
updates own state in response. 
Time complexity is measured by \emph{parallel time}:
the number of interactions divided by the population size $n$,
capturing the natural time scale in which each agent has $\Theta(1)$ interactions per unit time.
The agents collectively do a computation, 
e.g., population size counting: computing the value $n$. 
Counting is a fundamental task in distributed computing: 
knowing an estimate of $n$ often simplifies the design of protocols solving problems such as 
majority and leader election~\cite{AG15, berenbrink2018simpleLE,  gkasieniec2018almost,alistarh2017time, alistarh2018space, berenbrink18majority, bennun20majority, berenbrink_2020_optimalLE, sudo2018logarithmic, MocquardAABS2015,bilke2017brief}.

A protocol is defined by a \emph{transition function} with a pair of states as input and as output
(more generally to capture randomized protocols, 
a relation that can associate multiple outputs to the same input). 
For example, consider the simple counting protocol with transitions $L_i,L_j \to L_{i+j},F_{i+j}$, 
with every agent starting in $L_1$.
% , allows one agent (the one that never gets converted to $F$) to count the number of agents in the population by marking the counted ones as $F$. 
In population size $n$,
this protocol converges to a single agent in state $L_{n}$, 
with all other agents in state $F_i$ for some $i$.
The additional transitions $F_i,F_j \to F_j,F_j$ for $i<j$ propagate the output $n$ to all agents.
% This protocol \emph{stabilizes}, i.e., with probability 1, it reaches a configuration that is both correct and \emph{stable}: every subsequently reachable configuration is also correct.

% \noindent{\bf The dynamic size counting problem.}
% \noindent{\emph{Dynamic size counting}.}
\paragraph*{The dynamic size counting problem}
In contrast to most work, which assumes the population size $n$ is fixed over time,
we model an adversary that can add or remove agents arbitrarily and repeatedly during the computation.
All agents start in the same state, 
including newly added agents.
The goal is for each agent to approximately count the population size $n$,
which we define to mean that all agents should eventually store the same output $k$ in their states,
which with high probability is  within a constant multiplicative factor of $\log n$.\footnote{
    \emph{Nonuniform} protocols require agents to be initialized with an estimate $k$ of $\log n$ in order to accomplish other tasks, such as a ``leaderless phase clock''~\cite{alistarh2017time}.
    The bound $k=\Theta(\log n)$ is necessary and sufficient for correctness and speed in most cases~\cite{AG15, berenbrink2018simpleLE,  gkasieniec2018almost,alistarh2017time, alistarh2018space, berenbrink18majority, bennun20majority, berenbrink_2020_optimalLE, sudo2018logarithmic, MocquardAABS2015,bilke2017brief}.
}
Once all agents have the same output $k$,
they have \emph{converged}.
They maintain $k$ as the output for some time called the \emph{holding time} 
(after which they might alter $k$ even if the population size has not changed).
In response to a ``significant'' change in size from $n_{\mathrm{prev}}$ to $n_\mathrm{next}$,
agents should re-converge to a new output $k'$ of $\log n_\mathrm{next}$.
(Agents are not ``notified'' about the change; 
instead they must continually monitor the population to test whether their current output is accurate.)
Note that if $n_{\mathrm{prev}}$ is close to $n_\mathrm{next}$ (within a polynomial factor),
then $k$ may remain an accurate estimate of $n_\mathrm{next}$,
so agents may not re-converge in response to a small change.

Ideally the expected convergence time is small, and the expected holding time is large.
With a fixed size population, it is common to require the output to \emph{stabilize} to a value that never again changes after convergence, 
i.e., infinite holding time.
However, this turns out to be impossible with an adversary that can remove agents
(\cref{obs:impossible-to-stabilize}).
When changing from size $n_\mathrm{prev}$ to $n_\mathrm{next}$,
our protocol achieves expected convergence time $O(\log n_\mathrm{next} + \log \log n_\mathrm{prev})$
and expected holding time
$\Omega(n_\mathrm{next}^{c})$,
where $c$ can be made arbitrarily large.
The number of bits of memory used per agent is $O(\log^2 (s) + (\log \log n)^2)$, where $s$ is the maximum integer stored in the agents' memory after the change.  

While it is common to measure population protocol memory complexity by counting the number of states (which is exponentially larger than the number of bits required to represent each state),
that measure is a bit awkward here.
Our protocol is uniform---the same transition rules for every population size---so has an infinite number of producible states.
One could count expected number of states that will be produced, 
but this is a bit misleading:
in time $t$ each agent visits $O(t)$ states on average, so $O(t\cdot n)$ states total.
Counting how many bits are required is more accurate metric of the actual memory requirements.

%DD: seems redundant
% Assuming the adversary changes the population from $n$ to $n'$, the agents must handle a changed population size. 

% to a number 
% \todo{DD: it's not proportional; let's figure out how to use the language of $\log n$ rather than $n$ itself}
% proportional to the new population size $n'$ (e.g., $O(\log n')$) in time proportional to $n'$ (e.g., $O(\log n')$ time). 
% Additionally, we require the agents to maintain a 
% \todo{DD: we should be careful with the word ``stable'' since it usually means something different (i.e., permanent, as opposed to long holding time) ME: how about ``fixed?'' DD: I tried rewriting above.}
% fixed estimate if the adversary does not change the population. 

\paragraph*{The loosely-stabilizing counting problem}
% \noindent{\emph{Loosely-stabilizing counting}.}
The dynamic size counting problem has an equivalent characterization:
rather than removing agents and adding them with a fixed initial state,
the \emph{loosely-stabilizing} adversary sets each agent to an arbitrary initial state in a fixed-size population.
A protocol solves the dynamic size counting problem if and only if it solves the loosely-stabilizing counting problem,
with the same convergence and holding times
(\cref{lem:dynamic-size-equiv-loosely-stabilizing}).
Due to this equivalence, 
we analyze our protocol assuming a fixed population size and adversarial initial states.
In this case our convergence time $O(\log n + \log M)$ is measured as a function of the population size $n$ and the value $M$ that is the maximum $\est$ value stored in agents' memory. %size of integer fields in the adversarial initial states.
From the perspective of the dynamic size counting problem, these ``adversarial initial states'' would correspond to the agent states after correctly estimating the \emph{previous} population size, just prior to adding or removing agents.

\subsection{Related work}
\label{subsec:relatedworks}

% \paragraph*{Initialized counting with a fixed size population}\label{par:related-counting}
\noindent{\bf Initialized counting with a fixed size population.}
In population protocols with fixed size, 
there is work computing exactly or approximately the population size $n$.
% and approximate size counting problem, computing an approximation of $n$, e.g., $2^{\ceil{\log n}}$, have been considered. 
% We have summarized the related research in \cref{table:counting-summary}. 
For a full review 
% on size counting in population protocols, 
see~\cite{doty2021survey}.
Such protocols reach a \emph{stable} configuration from which the output cannot change.
Some of these counting protocols would still solve the counting problem in the presence of an adversary who can only add agents (see \cref{obs:stable-protocol-if-only-add-agents}). However, these protocols fail in the presence of an adversary who can also remove agents, since they work only in the initialized setting and rely on reaching a stable configuration (see \cref{obs:impossible-to-stabilize}). 

\noindent{\bf Self-stabilizing counting with a fixed size population.}
A population protocol is \emph{self-stabilizing} if,
from \emph{any} initial configuration,
it reaches to a correct stable configuration.
% In the self-stabilizing model, we are assuming there is an adversary who corrupts the memory of the agents and a self-stabilizing protocol should converge to the correct configuration once the adversary stops changing the agents' memories. 
% \todo{DD: I don't get the next statement; what makes our adversary more powerful? The self-stabilizing adversary can also cherry-pick an arbitrary initial configuration. ME: I had the same question and I thought you wrote this part! Haha}
% Note that our dynamic adversary is even a more powerful adversary since it can start with a very large population size and cherry-pick any arbitrary initial configuration by removing the agents. 
Self-stabilizing size counting has been studied~\cite{AspnesBBS2016, beauquier2015space, BeauquierSS2007, IZUMI2014},
but provably requires adding a ``base station'' agent that cannot be corrupted by the adversary. 
In these protocols the base station is the only agent required to learn the population size.
Aspnes, Beauquier, Burman, and Sohier~\cite{AspnesBBS2016} showed a time- and space-optimal protocol  that solves the exact counting problem in $O(n \log n)$ time, 
using $1$-bit memory for each non-base station agent.

% \paragraph*{Size regulation in a dynamically sized population}\label{par:related-goldwasser-size-regulation}
\noindent{\bf Size regulation in a dynamically sized population.}
The model described by Goldwasser, Ostrovsky, Scafuro, and Sealfon~\cite{goldwasser2018population} is close to our setting. 
% DD: redundant with next sentence
% They study population protocols in the presence of an adversary who can add or remove agents.
They consider the \emph{size regulation problem}:
approximately maintaining a target size (hard-coded into each agent) using $O(\log \log n)$ bits of memory per agent,
despite an adversary that 
(like ours) adds or removes agents.
That paper assumes a model variation in which: 
\begin{itemize}
    \item
    The agents can replicate or self-destruct.
    
    \item
    The computation happens through synchronized rounds of interactions. At each round the scheduler selects a random matching of size $k=O(n)$ agents to interact. 
    
    \item
    The adversary's changes to the population size are limited. The adversary can insert or delete a total of $o(n^{1/4})$ agents within each round.
\end{itemize}

The latter two model differences above crucially rule out their protocol as useful for our problem.
We use the standard asynchronous scheduler,
and much of the complexity of our protocol is to handle drastic population size changes
(e.g., removing $n-\log n$ agents).
Additionally, their protocol heavily relies on flipping coins of 
% \todo{DD: It's not clear why this last thing is a problem. (Answer: because we don't have an estimate of $n$ yet.)}
bias $\frac{1}{\sqrt{n}}$ that we cannot utilize since the agents don't start with an estimate of $n$. Moreover, even when the agents compute their estimate, the population size might change.
\noindent{\bf Loosely-stabilizing leader election.}\label{par:related-LS-LE}
Sudo, Nakamura, Yamauchi, Ooshita, Kakugawa, and Masuzawa~\cite{sudo2010looselystabilizing} introduce loose-stabilization as a relaxation for the self-stabilizing leader election problem in which the agents must know the exact population size to elect a leader. The loosely stabilizing leader election guarantees that starting from any configuration, the population will elect a leader within a short time. After that, the agents hold the leader for a long time but not forever (in contrast with self-stabilization). On the positive side, the agents no longer need to know the exact population size to solve the loosely-stabilizing leader election, but a rough upper bound suffices. 
Loosely-stabilizing leader election has been studied, providing a time-optimal protocol that solves the leader election problem~\cite{Sudo2021TimeoptimalLL} and a tradeoff between the holding and convergence times~\cite{izumi16,sudo2020}. 
% --> if the expected holding time of an LS-LE protocol is $\frac{\beta}{n}$, its expected convergence time must be $\Omega(\log \beta)$. 

% \paragraph*{Computation with dynamically changing inputs}
\noindent{\bf Computation with dynamically changing inputs.}
Alistarh, Töpfer, and Uznański~\cite{alistarh2021robustComparision} consider the dynamic variant of the comparison problem. In the comparision problem, a subset of population are in the input states $X$ and $Y$ and the goal is to compute if $X > Y$ or $X < Y$. In the dynamic variant of the comparision problem, they assume an adversary who can change the counts of the input states at any time. The agents should compute the output as long as the counts remain untouched for sufficiently long time. They propose a protocol that solves the comparision problem in $O(\log n)$ time using $O(\log n)$ states per agent, assuming $|X| \geq C_2\cdot|Y| \geq C_1 \log n$ for some constants $C_1, C_2> 1$. 

% \todo{DD: I'd use the same opinions as above (let's call them $A,B$ or $X,Y$ in both places.}
Berenbrink, Biermeier, Hahn, and Kaaser~\cite{berenbrink2021selfstabilizing} consider the adaptive majority problem (generalization of the comparison problem~\cite{alistarh2021robustComparision}). 
At any time every agent has an opinion from $\{X, Y\}$ or undecided and their opinions might change adversarially. The goal is to have agreement in the population about the majority opinion. They introduce a non-uniform loosely-stabilizing leaderless phase clock that that uses $O(\log n)$ states to solve the adaptive majority problem. This is similar to having an adversary who can add or remove agents with different opinion.
However, all agents are assumed already to have an estimate of $\log n$ that remains untouched. 
Thus it is not straightforward to use their protocol to solve our problem of obtaining this estimate.

% \url{https://arxiv.org/abs/2003.06485}
% \url{https://arxiv.org/abs/2106.13002}

\section{Definitions and notation}
\label{subsec:defs}
% \emph{uniform protocol}

% from majority paper: 
% In this paper we deal with \emph{nonuniform} protocols in which a different $\Lambda$ and $\delta$ are allowed for different population sizes $n$ 
% (one for each possible value of $\ceil{\log n}$),
% but we abuse terminology and refer to the whole family as a single protocol.
% In all cases (as with similar nonuniform protocols),
% the nonuniformity is used to embed the value $\ceil{\log n}$ into each agent; 
% the transitions are otherwise ``uniformly specified''.

A \emph{population protocol} is a pair $\calP=(\Lambda,\Delta)$,
where $\Lambda$ is a finite set of \emph{states}, and $\Delta \subseteq (\Lambda \times \Lambda) \times (\Lambda \times \Lambda)$ is the \emph{transition relation}.
(Often this is defined as a function $\delta: \Lambda \times \Lambda \to \Lambda \times \Lambda$, but we allow randomized transitions, where the same pair of inputs can randomly choose among multiple outputs.)

A \emph{configuration} $\vc$ of a population protocol is a multiset over $\Lambda$ of size $n$, giving the states of the $n$ agents in the population. For a state $s \in \Lambda$, we write $\vc(s)$ to denote the count of agents in state $s$.
A \emph{transition} 
is a 4-tuple, 
written $\alpha: r_1,r_2 \to p_1,p_2$,
such that $((r_1,r_2), (p_1,p_2)) \in \Delta$.
If an agent in state $r_1$ interacts with an agent in state $r_2$, then they can change states to $p_1$ and $p_2$.
This notation omits explicit probabilities;
our main protocol's transitions can be implemented so as to always have either one or two possible outputs for any input pair,
with probability $1/2$ of each output in the latter case.\footnote{For the purpose of representation, we make an exception in our protocol, when we show agents generate a geometric random variable in one line (see \cref{protocol:timer-routine}). However, we can assume a geometric random variable is generated through $O(\log n)$ consecutive interactions with each selecting out of two possible outputs ($\mathbf{H}$ or $\mathbf{T}$).}
% \todo{DD: TODO: justify that this is the case for when we need a geometric, that it can be done by successive fair coin flips}
For every pair of states $r_1,r_2$ without an explicitly listed transition $r_1,r_2 \to p_1,p_2$, there is an implicit \emph{null} transition $r_1,r_2 \to r_1,r_2$ in which the agents interact but do not change state.
For our main protocol,
we specify transitions formally with pseudocode that indicate how agents alter each independent field in their state.
We say a configuration $\vec{d}$ is \emph{reachable} from a configuration $\vec{c}$ if applying 0 or more transitions to $\vec{c}$ results in $\vec{d}$.

When discussing random events in a protocol of population size $n$,
we say event $E$ happens 
\emph{with high probability} 
if $\Pr{\neg E} = O(n^{-c})$,
where $c$ is a constant that depends on our choice of parameters in the protocol,
where $c$ can be made arbitrarily large by changing the parameters.
For concreteness, we will write a particular polynomial probability such as $O(n^{-2})$,
but in each case we could tune some parameter 
(say, increasing the time complexity by a constant factor)
to increase the polynomial's exponent.

To measure time we count the total number of interactions (including null transitions such as $a,b \to a,b$ in which the agents interact but do not change state), and divide by the number of agents $n$.

In a \emph{uniform} protocol 
(such as the main one of this paper), 
the transitions are independent from the population size $n$ (see~\cite{DEMST18} for a formal definition). In other words, a single protocol computes the output correctly when applied on any population size. 
In contrast, in a nonuniform protocol
different transitions are applied for different population sizes. 

% \emph{stabilizing protocol}
A protocol \emph{stably} solves a problem if the agents eventually reach a correct configuration with probability $1$, and no subsequent interactions can move the agents to an incorrect configuration; i.e., the configuration is \emph{stable}. A population protocol is self-stabilizing if from any initial configuration, the agents stably solve the problem.

\section{Dynamic size counting}

% \todo{ME: having a correct estimate and the rest of the memory as garbage, has a very small holding time and it should not be considered in the set of $C(n,\delta)$.}

In a population of size $n$, 
define $C(n,\delta)$ to be the set of correct configurations $\vec{c}$ such that every agent $u$ in $\vec{c}$ obeys 
$(1-\delta)\log n < u.\est < (1+\delta)\log n$.
% \todo{DD: $f_h$ is not well-defined. It could depend on which configuration in $C(n,\delta)$ we start in.}
Let $t_h$  be any time bound. Moreover, we define $L(n,\delta, t_h) \subset C(n,\delta)$ the subset of correct configurations 
such that as the expected time for protocol $P$ starting from a configuration $\vec{l} \in L(n,\delta, t_h)$ to stay in $C(n,\delta)$ is at least $t_h(n)$.

% with holding time (i.e., remain in $C(n, \delta)$) for $t_h$ time. Meaning that if the population starts in a configuration 
% \me{Given a correct configuration $\vec{c} \in C(n,\delta)$,} we define \mes{$f_h(C(n,\delta))$}{$f_h(\vec{c}, C(n,\delta))$ as the expected time for protocol $P$ \me{starting from $\vec{c}$} to stay in $C(n,\delta)$. }
% We also use $R(n)$ to denote the set of all reachable configuration with population size $n$.

% Let $n_{\mathrm{prev}}$ and $n_{\mathrm{next}}$ denote the population size before and after the change of the adversary. We define $f\qty(\vec{r}, \vec{c})$ as the expected time to reach from configuration $\vec{r}$ to $\vec{c}$.
\begin{definition}
     Let $n_{\mathrm{prev}}$ and $n_{\mathrm{next}}$ denote the previous and next population size. 
     A protocol $P$ 
    %  $\qty(t(n_{\mathrm{prev}}, n_{\mathrm{next}}), \delta )$ 
     solves the \emph{dynamic size counting} problem if there is a $\delta>0$,
     called the \emph{accuracy},
     such that
     if the population size changes from $n_\mathrm{prev}$ to $n_\mathrm{next}$,
     the protocol reaches a configuration $\vec{l}$ in $L(n, \delta,t_h)$ with high probability.
     The time needed to do this is called the \emph{convergence} time. Moreover, $t_h$, the time that the population stays in $C(n_\mathrm{next},\delta)$, is called the \emph{holding} time.
\end{definition}
%     %  defining for all $n \in \N^+$,
%     %  $C(n) = \{ \vec{c} \mid (\forall u \in \calA)\ (1-\delta)\log (n) < u.\est < (1+\delta)\log n \}$,
     
%     %  \begin{enumerate}
%         % \item 
%         for every $n_{\mathrm{prev}}, n_{\mathrm{next}} \in \N^+$, $\forall \vec{r} \in R(n_{\mathrm{prev}})$, $\exists \vec{c} \in C(n_{\mathrm{next}})$ such that $f\qty(\vec{r}, \vec{c}) \leq t(n_{\mathrm{prev}}, n_{\mathrm{next}})$.
%         % \item $\forall \vec{c} \in C(n_{\mathrm{next}})$ and $\forall u \in \calA $ we have $(1-\delta)\log (n_{\mathrm{next}}) < u.\est < (1+\delta)\log n$.
%     % \end{enumerate}
    
%         $f_h(C(n_{\mathrm{next}})) \geq t_h(n_{\mathrm{next}})$ 
    
% \end{definition}

% \emph{loosely-stabilizing protocol}
A population protocol is \emph{$\qty(t_c(n),t_h(n))$-loosely stabilizing} if starting from any initial configuration, the agents reach a correct configuration in $t_c(n)$ time and stay in the correct configuration for additional $t_h(n)$ time~\cite{sudo2010looselystabilizing,Sudo2021TimeoptimalLL}. 
% \todo{DD: add citation to self-stabilizing papers}
In contrast to self-stabilizing~\cite{AAFJ08, burman2021timeoptimal}, subsequent interactions can move the agents to an incorrect configuration; however, the agents recover quickly from an incorrect configuration.

% We call a set of configurations $C$ correct if $(1-\delta)\log n < u.\est < (1+\delta)\log n$ for all agents $u \in \calA$ and all $\vec{c} \in C$. 
Given any starting configuration $\vec{s} \not \in C(n,\delta)$ of size $n$, we define $f_c(\vec{s}, L(n, \delta,t_h))$ as the expected time to reach a correct configuration in $L(n, \delta,t_h)$. 
% Additionally, we define $M$ as the largest integer value the agents stored in $\vec{s}$.
% \todo{DD: $f_h(C)$ could be used in the previous definition also, so should be moved up there.}

\begin{definition}~\cite[Definition 2]{sudo2010looselystabilizing}
    Let $t_c(n,M)$ and $t_h(n)$ be functions of $n$, the largest integer value $M$ in the initial configuration $\vec{s}$, and the set of correct configuration $C(n,\delta)$. A protocol $P$ is a $t_c(n, M), t_h\qty(n) , \delta $ loosely-stabilizing population size counting protocol if there exists a set $L(n, \delta,t_h)$ of configurations satisfying:
    % \todo{Change $t_c(n)$ to $t_c(n,M)$ where $M$ is maximum integer that adversary put in agents' memory}
    
    For every $n$ and every initial configuration $\vec{s} \not \in C(n,\delta)$ of size $n$, $f_c(\vec{s}, L(n, \delta,t_h)) \leq t_c(n, M)$
        % \item \mes{$f_h(C(n_\mathrm{next}, \delta))$}{$\min_{\vec{c} \in C(n,\delta)}{f_h(\vec{c}, C(n,\delta))}$} $\geq t_h(n)$ 
        % \item $\forall \vec{c} \in C(n,\delta)$ and $\forall u \in \calA $ we have $(1-\delta)\log n < u.\est < (1+\delta)\log n$.
\end{definition}

\subsection{Basic properties of the dynamic size counting problem}

We first observe that the key challenge in dynamic size counting is that the adversary may remove agents.
If the adversary can only add agents,
the problem is straightforward to solve with optimal convergence and holding times.

\begin{observation}
    \label{obs:stable-protocol-if-only-add-agents}
    Suppose the adversary in the dynamic size counting problem only adds agents.
    Then there is a protocol solving dynamic size counting with $O(\log n)$ convergence time (in expectation and with probability $\geq 1 - O(1/n)$) and infinite holding time.
\end{observation}

\begin{proof}
    Each agent in the initial state $s$ generates a geometric random variable.
    After the last time that the adversary adds agents, resulting in $n$ total agents,
    exactly $n$ geometric random variables will have been generated.
    Agents propagate the maximum by epidemic using transition $a,b \to \max(a,b), \max(a,b)$,
    taking $3 \ln n$ time to reach all agents
    with probability
    
    $\geq 1-\frac{1}{n^2}$~\cite[Corollary 2.8]{burman2021timeoptimal}.
    % exact probability bound: $2.5 \ln n n^{-3} \leq n^{-2}
    The maximum of $n$ i.i.d.~geometric random variables is in the range 
    $[\log n - \log \ln n, 2\log n]$ 
    with probability 
    $\geq 1-\frac{1}{n}$~\cite[Lemma D.7]{doty2018efficient}.
    % exact probability bound: n^{-1}
\end{proof}

In contrast,
if the adversary can \emph{remove} agents,
then even if it is guaranteed to do this exactly once, 
no protocol can be stabilizing,
i.e., have infinite holding time.
\todo{We can use a similar argument to loosely stabilizing leader election and add a polynomial bound on this.. Let's say there is a subpopulation of size n/10 and we can calculate the probability that all the interactions happen within this subpopulation }
\begin{observation}
    \label{obs:impossible-to-stabilize}
    Suppose the adversary in the dynamic size counting problem will remove agents exactly once.
    Then any protocol solving the problem has finite holding time.
\end{observation}

\begin{proof}
    Suppose otherwise.
    Let the initial population size be $n$ and the later size be $n' < n$.
    The protocol must handle the case where the adversary \emph{never} removes agents,
    since in population size $n$ this is equivalent to an adversary who starts with $n+1$ agents and immediately removes one of them.
    Thus if the adversary waits sufficiently long before the removal,
    then all agents stabilize to output $k = \Theta(\log n)$.
    In other words,
    no sequence of transitions can alter the value,
    including transitions occurring only among any subpopulation of size $n'$.
    So after the adversary removes $n-n'$ agents,
    the remaining $n'$ agents are unable to alter the output $k$,
    a contradiction if $n'$ is sufficiently \todo{Reviewer: Please explain how small and why} small compared to $n$.
\end{proof}

\cref{lem:dynamic-size-equiv-loosely-stabilizing} shows that the dynamic size counting problem is equivalent to the loosely-stabilizing counting problem.
Due to this equivalence, 
our correctness proofs will use the loosely-stabilizing characterization.

\begin{toappendix}
To prove \cref{lem:dynamic-size-equiv-loosely-stabilizing}, 
we require the following result, 
proven in~\cite[Lemma 4.2]{doty2018efficient}.
It states that for any finite set $\Lambda_f$ of states producible in a protocol from a uniform initial state,
for any sufficiently large initial population size $n$,
WHP many copies of each state in $\Lambda_f$ appear.\footnote{
    Lemma 4.2 in \cite{doty2018efficient} is stated slightly differently.
    Rather than an arbitrary finite set of states $\Lambda_f$,
    it considers for some fixed $m \in \N^+$ and $0 < \rho \leq 1$, 
    the set of states producible using $m$ different types of transitions,
    each having transition probability at least $\rho$
    (necessarily a finite set for fixed $m$ and $\rho$).
    Setting $m$ to the maximum number of types of transitions needed to produce any state $q \in \Lambda_f$,
    and $\rho$ to be the minimum transition probability among any of those transitions,
    we obtain the simpler lemma statement used here.
    Also,
    Lemma 4.2 in \cite{doty2018efficient} allowed more general initial configurations, permitting agents in different states, but each having count $\Omega(n)$ (so-called ``dense'' configurations).
}

\begin{lemma}[\hspace{-0.01cm}\cite{doty2018efficient}]
    \label{lem:density}
    Let $P$ be a population protocol,
    and let $\Lambda_f$ be a finite set of states,
    each producible in $P$ from sufficiently many agents in state $s$.
    Then there are constants $\epsilon,\delta,n_0 > 0$ such that, 
    for all $n \geq n_0$,
    starting from $n$ agents in state $s$,
    if $\vec{d}$ is the configuration reached at time $1$,
    then
    $
        \Pr{(\forall q \in \Lambda_f)\ \vec{d}(q) \geq \delta n} 
        \geq 1 - 2^{-\epsilon n}.
    $
\end{lemma}
\end{toappendix}

Recall that we define $M$ as the largest integer value the agents stored in the starting configuration $\vec{s}$.

\todoi{The following lemma is proven in the full version of this paper. [add citation?]}
\begin{lemmarep}
    \label{lem:dynamic-size-equiv-loosely-stabilizing}
    A protocol solves the dynamic size counting problem with convergence time $t_c(n, M)$
    % \todo{ME: In our formal definition, we used $t_c(n,s)$ in which $s$ is the initial configuration. DD; fixed by not mentioning what are the input(s) to the time bounds.} 
    and holding time $t_h(n)$ if and only if it solves the loosely-stabilizing counting problem with convergence time $t_c(n,M)$ and holding time $t_h(n)$.
\end{lemmarep}

\begin{proofsketch}
    Any states present in an adversarially prepared configuration $\vec{c}$ will be produced in large quantities from any sufficiently large initial configuration of all initialized states $s$~\cite[Lemma 4.2]{doty2018efficient}.
    The dynamic size adversary can then remove agents to result in $\vec{c}$,
    which the protocol must handle,
    showing it can handle an arbitrary initial configuration. 
\end{proofsketch}

\begin{proof}
    % We show how an adversary with the ``dynamic size'' power can simulate the adversary with the ``loosely stabilizing'' power, and vice versa.
    The easy direction is that if the protocol $P$ solves the loosely-stabilizing counting problem,
    then it solves the dynamic size counting problem. Starting from any possible configuration $\vec{s}$ of size $n$, with largest integer $M$,
    $P$ converges 
    \todo{DD: specify inputs to time bounds.}
    in time $t_c(n, M)$ 
    with holding time $t_h(n)$.
    The configuration immediately after the dynamic size adversary adds or removes agents is simply one of the configurations that $P$ is able to handle.
    
    To see the reverse direction,
    suppose $P$ solves the dynamic size counting problem;
    we argue that $P$ also solves the loosely-stabilizing counting problem.
    Let $\vec{c}$ be any configuration of $P$.
    The dynamic size adversary can do the following to reach configuration $\vec{c}$.
    Let $\Lambda_f = \{ q \mid \vec{c}(q) > 0\}$ be the set of states in $\vec{c}.$
    Apply \cref{lem:density} with this choice of $\Lambda_f$.
    Choose $n$ sufficiently large that $\delta n \geq \max_{q \in \Lambda_f} \vec{c}(q)$.
    Then \cref{lem:density} says that starting from $n$ agents in state $s$, with probability at least $1 - 2^{-\epsilon n}$,
    in the configuration $\vec{d}$ at time 1,
    for all $q \in \Lambda_f$,
    $\vec{d}(q) \geq \delta n \geq \vec{c}(q)$.
    Thus $\vec{d} \geq \vec{c}$.
    Now the adversary removes agents from $\vec{d}$ to result in $\vec{c}$.
    Let $n' = \| \vec{c} \|$.
    \todo{DD: specify inputs to time bounds.}
    The protocol
    converges from $\vec{c}$ with largest integer $M'$ in time $t_c(n, M')$, with holding time $t_h(n)$.
    But since $\vec{c}$ is an arbitrary configuration of the protocol,
    this implies that $P$ solves the loosely-stabilizing counting problem.
\end{proof}

\subsection{High-level overview of dynamic size counting protocol}

In this section we briefly describe our protocol solving the dynamic size counting,
which is defined formally in \cref{subsec:protocol-formal-definition}.
By \cref{lem:dynamic-size-equiv-loosely-stabilizing} it suffices to design a protocol solving the loosely-stabilizing counting problem for a fixed population size $n$.
Our protocol uses the ``detection'' protocol of~\cite{alistarh17detection}. 
Consider a subset of states designated as a ``source''.
A detection protocol alerts all agents whether a source state is present in the population. 

In \cref{protocol:dynamic-counting-all}, 
the population maintains
% \todo{DD: Let's change this to ``several'' dynamic groups; it's not $\log n$ unless things have converged.}
several dynamic \emph{groups},
with the agent's group stored as a positive integer field \group. The \group\ values are not fixed:
each agent changes its \group\ field on every interaction,
with equal probability either incrementing \group\ or setting it to 1.
%DD: next sentence seems redundant with previous
% Starting from an arbitrary initial $\group$ values, the agents update their $\group$ value immediately. 
We show that,
no matter the initial group values, 
after $O(\log n)$ time the group values will be in the range $[1, 8\log n]$ WHP. 
The distribution of \group\ values is very close to that of $n$ i.i.d.~geometric random variables, in the sense that each agent's \group\ value is independent of every other,
with expected $n/2^i$ agents having $\group = i$ if each agent has had at least $i$ interactions.\footnote{
    The difference is that a geometric random variable $G$ obeys $\Pr{G=j} = 1/2^j$ for \emph{all} $j \in \N^+$,
    but after $i$ interactions an agent $u$ can increment $u.\group$ by at most $i$,
    so $\Pr{u.\group=j}=0$ if $j \gg i$.
}

The agents store an array of ``signal'' integers in their $\signalArray$ field,
as a way to track the existing \group\ values in the population. 
Each agent in the $i$'th group is responsible to \emph{boost} the signal associated with $i$.
The goal is to have $\signalArray[i] > 0$ for all agents if and only if some agent has $\group=i$.

The detection protocol of~\cite{alistarh17detection},
explained below,
provides a technique for agents to know which groups are still present.
Once a signal for group $k$ fades out, the agents speculate that there is no agent with $\group = k$. 
Depending on the current value stored as $\est$ in agents' memory and the value $k$, this might cause re-calculating the population size. 
The agents are constantly checking for the changes in the $\signalArray$.
% for change in 
% \todo{DD: this is a confusing statement; no agent stores ``the set of existing \group\ values'' directly. What they are checking is the \signalArray.}
% the set of existing $\group$ values; 
They re-compute $\est$ once there is a large gap between $\est$ and the first 
% \todo{DD: We bounce between the terms ``index'' and ``group''; let's just always use the word ``group''.
% ME: I changed it here, but sometimes we are referring to indices of the $\signalArray$ directly. I did not change those instances.}
\group\ $i$ with $\signalArray[i] = 0 $. We call $i$ 
the \emph{first missing value} 
(stored in the field $\missingV$).
% DD: The next sentence seemed a bit misleading; the *whole* protocol is loosely stabilizing, not just this part of it.
% Our protocol is robust to adversarial changes such as when the adversary adds agents with arbitrary large $\group$ values.

% \todo{DD: This next paragraph is my attempt to indicate something novel we are doing with the detection scheme. It goes on a bit longer than I intended, but I think this is good to avoid the charge that we are merely using some other protocol, with some other analysis, and don't need to do any interesting analysis ourselves.}
The \signalArray\ array is updated as follows.
% Direct contact with a
An agent with $\group=k$ 
% (including the same agent generating group $k$ itself) 
sets $\signalArray[k]$ to its maximum possible value 
($3k+1$);
we call this \emph{boosting}.
Other groups $k$ are updated between two agents $u,v$ with $u.\signalArray[k] = a$ and $v.\signalArray[k] = b$ via \emph{propagation} transitions that set both agent's $\signalArray[k]$ to $\max(a-1,b-1,0)$.
The paper~\cite{alistarh17detection} used a nonuniform protocol where each agent \emph{already} has an estimate of $\log n$.
They prove that 
if the state being detected 
(in our case, a state with $\group=k$) is \emph{absent} and the current maximum signal is $c$,
then all agents will have signal 0 within $\Theta(c)$ time.
However, if the state being detected is present,
then the boosting transitions 
(occurring every $O(1)$ units of parallel time on average in the worst case that its count is only 1)
will keep the signal positive in all agents with high probability.
For this to hold,
it is critical that the maximum value set during boosting is $\Omega(\log n)$;
the nonuniform protocol of~\cite{alistarh17detection} uses its estimate of $\log n$ for this purpose.

Crucially,
our protocol associates smaller maximum signal values to smaller group values
(so many are much smaller than $\log n$),
to ensure that a signal does not take abnormally long to get to 0 when its associated group value is missing.
Otherwise, if we set each signal value to $\Omega(\log n)$ (based on the agent's current estimate \est\ of $\log n$) during boosting,
then it would take time proportional to \est\ 
(which could be much larger than the actual value of $\log n$) to detect the absence of a \group\ value.
Thus is it critical that we provide a novel analysis of the detection protocol,
showing that the signals for smaller group values $k \ll \log n$ remain present with high probability.
This requires arguing that the boosting reactions for such smaller values are happening with sufficiently higher frequency,
due to the higher count of agents with $\group=k$,
to compensate for the smaller boosting signal values they use.

\subsection{Formal description of loosely-stabilizing counting protocol}
\label{subsec:protocol-formal-definition}

The $\dcountingA$ protocol 
(\cref{protocol:dynamic-counting-all}) 
divides agents among $\Theta(\log n)$ groups via the $\updateGroup$ subprotocol. 
% Regardless of what $\group$ value the agents are initialized with, 
The agents update their $\group$ from $i$ to $i+1$ with probability $1/2$ or reset to group $1$ with probability $1/2$. 
The number of agents at each group and the total number of groups are both random variables that are dynamically changing through time.
We show that the total number of groups remains close to $\log n$ at all times with high probability.

The agents start with arbitrary (or even adversarial) $\group$ values but we show that WHP the set of $\group$ values will converge to $[1, 8 \log n]$ within $O(\log n)$ time (\cref{cor:dynamic-groups}). 
Additionally, each agent stores an array of $O(\log n)$ signal values in their $\signalArray$ field. 
It is crucial for agents to maintain positive values in the $\signalArray[i]$ if some agent has $\group = i$. 
They use the first \group\ $i$ with $\signalArray[i] = 0$ (stored in $\missingV$). 
The agents use $\missingV$ as an approximation of $\log n$ and constantly compare it with their $\est$ value. 

Depending on the $\est$ value stored in agents' memory, the agents maintain 3 main phases of computation: 

\begin{description}
    \item[$\nPhase$:] 
    An agent stays in the $\nPhase$ as long as there is a small gap between $\est$ and $\missingV$: 
    $0.25 \cdot \est \leq \missingV \leq 2.5 \est$.
    % In other words, an agent maintains the $\nPhase$ if all the values stored in its memory are proportional.
    
    \item[$\wPhase$:] 
    An agent switches from $\nPhase$ to $\wPhase$ if it sees a large gap between the $\missingV$ and $\est$: 
    $\missingV  \not \in \{0.25 \cdot \est , \ldots, 2.5 \cdot \est\}$.
    The purpose of $\wPhase$ is to give enough time to the other agents so that by the end of the $\wPhase$ for one agent, 
    with high probability every other agent has also noticed the large gap between the $\missingV$ and $\est$ and entered $\wPhase$.
    
    \item[$\cPhase$:] 
    During the $\cPhase$, every agent use a new geometric random variable and propagates the maximum by epidemic. We set $\wPhase$ long enough so that with high probability when the first agent switches to the $\cPhase$, the rest of the population are all in $\wPhase$. By the end of $\cPhase$, every agent switches back to $\nPhase$.
\end{description}

Below we explain each subprotocol in more detail.

% \todo{DD: This explanation uses the dynamic size adversary; change it to use the loosely-stabilizing adversary.}
% Note on the dynamic size counting: The total number of groups changes as the adversary changes (increase/decrease) the population size. Assuming the adversary decrements the population size from $n$ to $n'$ ($n' < n$); then, the set of existing $\group$ values should change from $[1- c \cdot \log n]$ to $[1- c \cdot \log n']$ for some constant $c > 1$. Similarly, in the case that the adversary increases the population size from $n$ to $n''$, new $\group$ values appear and the set of existing $\group$ values changes from $[1- c \cdot \log n]$ to $[1- c \cdot \log n'']$.
% Thus, the agents observe the appearance of new groups or the departure of existing groups which eventually guide them to re-calculate their estimation of the population size.

% \todo{DD: it was confusing to have \agent\ as both a type declaration in the protocol and also a local variable in a loop. I removed the type declaration.}
% \todo{DD: I like the convention where only fields are in {\tt TrueType font}, to make it easy to see when a field is being written or read. Can we make other things (e.g., parameters/local variables $u,v,agent$, and subroutine names $\mathrm{DynamicCounting}$) a different font? Done.}
\begin{algorithm}[ht]
    \floatname{algorithm}{Protocol}
%   \caption{$\dcountingA(\agent\ \au, \agent\ \av)$
    \caption{$\dcountingA(\au, \av)$
%   \\
%       \textbf{Initialization: }\\
%       $\au.\group \gets 1$; 
%       $\au.\RG \gets 0$\\
%       $\au.\phase \gets \nPhase$\\
%       $\au.\signalArray \gets []$\\
%       $\au.\timer \gets 0$
%       $\au.\missingV \gets 0$;
%       $\au.\est \gets 0$;
    }
    \label{protocol:dynamic-counting-all}
    \begin{algorithmic}[100]
        \For{$\agent \in \{\au, \av\}$}
            \State {$\updateGroup(\agent)$}
        \EndFor
        \State {$\resetCountDown(\au, \av)$}
        \For{$\agent \in \{\au, \av\}$}
            \State {$\updateMissingV(\agent)$}
            \State {$\checkSizeChange(\agent)$}
            \If{$\agent.\phase \neq \nPhase$}
                \State {$\timerRoutine(\agent)$}
            \EndIf
        \EndFor
        \State {$\EpidemicMaxRG(\au, \av)$}
        \For{$\agent \in \{\au, \av\}$}
            \If{$\agent.\phase = \nPhase$}
                \State{$\agent.\est \gets \agent.\RG $ }
            \EndIf
        \EndFor
    \end{algorithmic}
\end{algorithm}

In every interaction, both sender and receiver update their group according to the rules of the $\updateGroup$ subprotocol.
If we look at the distribution of the $\group$ values after $O(\log n)$ time, there are about $n/2$ agents in group $1$, $n/4$ agents in group $2$ and $n/2^i$ agents in group $i$ (see \cref{fig:group-dist}). 
Note that the number of agents in each group decreases exponentially, but we ensure that agents with larger $\group$ values use stronger signals to propagate, 
since there is less support for those groups. 

\begin{algorithm}[ht]
    \floatname{algorithm}{Protocol}
    \caption{$\updateGroup(\agent\ \au)$
       % \\
       % \textbf{Initialization: }\\
       % $\au.\group \gets 1$
    }
    \label{protocol:stationary-geometric}
    \begin{algorithmic}[100]
        \State {$\au.\group \gets 
        \begin{cases}
            \au.\group+1 
            & \text{with probability } 1/2
        \\ 
            1 
            & \text{with probability } 1/2
        \end{cases}
        $
        }
        % \au.\group+1$ with probability $1/2$}
       % \State {$\au.\group \gets 1$ with probability $1/2$}
    \end{algorithmic}
\end{algorithm}

To notify all agents about the set of all $\group$ values that are generated among the population, we use the detection protocol of~\cite{alistarh17detection} that is also used as a synchronization scheme in~\cite{berenbrink2021selfstabilizing}. 
The agents store an integer for each $\group$ value that is generated by the population. The $\signalArray$ is an array of length $\Theta(\log n)$ such that a positive value in index $i$ represents some agents in the population have generated $\group =i$. 
Note that, as an agent updates its $\group$, it boosts multiple signals based on its $\group$ value, e.g., an agent with $\group = i$ helped boosting all the indices $1, 2, 3, \ldots, i$ of $\signalArray$ in its last $i$ interactions. We use the $\resetCountDown$ protocol to keep the signal of group $i$ positive as long as some agents have generated $\group = i$. 
% Keep in mind that, when the adversary changes the population size, the set of all groups will change too.
% In the case that the population size increases, hence, increasing the total number of groups, the agents can easily observe the appearance of more groups and their associated signals. However, this is not as easy for the case when the adversary removes agents. 

\begin{algorithm}[ht]
    \floatname{Algorithm}{Protocol}
    \caption{$\resetCountDown(\agent\ \au,\agent\ \av)$}
    \label{protocol:reset-count-down}
    \begin{algorithmic}[100]
        \LeftComment{Boosting:}
        \State {$\au.\signalArray[\au.\group] \gets (3 \cdot \au.\group ) + 1$}
        \State {$\av.\signalArray[\av.\group] \gets (3 \cdot \av.\group ) + 1$}

        \LeftComment{Propagate signal:}
        \For{$i \in \{1,2,\ldots, \texttt{Max}(|\au.\signalArray|, |\av.\signalArray| \}$}
            \State{$m \gets \texttt{Max($\au.\signalArray[i], \av.\signalArray[i]$)}$}
            \State{$\au.\signalArray[i], \av.\signalArray[i] \gets \texttt{Max}(0, m - 1)$}
        \EndFor
    \end{algorithmic}
\end{algorithm}

Regardless of the initial configuration, the distribution of $\group$ values changes immediately (in $O(\log n)$ time) but it might take more time for the $\signalArray$ to get updated. It takes $O(i)$ time for $\signalArray[i]$ to hit zero. The larger the index $i$, $\signalArray[i]$ leaves the population slower. Hence, the agents look at the \emph{first missing signal} that they observe among the array of all signals. 

\begin{algorithm}[ht]
    \floatname{Algorithm}{Protocol}
    \caption{$\updateMissingV(\agent\ \au)$}
    \label{protocol:update-missing-value}
    \begin{algorithmic}[100]
        \LeftComment{Finds the first appearance of a zero in $u.\signalArray$ beyond index $\ceil{\texttt{log}(\au.\est)}$}
        \State{$s \gets \ceil{\texttt{log}(\au.\est)}$}
        \State {$\au.\missingV \gets \min \{i \in [s, |\au.\signalArray|] \mid \au.\signalArray[i] = 0\}$}
       % \State {$\au.\missingV \gets |\au.\signalArray|$}
       % \State{$s \gets \ceil{\texttt{log}(\au.\est)}$}
       %  \For{$i$ in range ($s, |\au.\signalArray|$)}
    %       \If{$\au.\signalArray[i] = 0$}
    %           \State {$\au.\missingV \gets i$}   
       %     \EndIf
    %     \EndFor
    \end{algorithmic}
\end{algorithm}

Once there is a large gap between the first missing $\group$ ($\missingV$) and the agents' estimation of $\log n$ ($\est$), each agent individually moves to a \emph{waiting phase} and waits for other agents to catch the same gap between their $\est$ and $\missingV$. Note that, we time this phase as a function of $\missingV$ and not the $\est$ since the $\est$ is not valid anymore and might be much smaller or larger that the true value of $\log n$.

\begin{algorithm}[ht]
    \floatname{Algorithm}{Protocol}
    \caption{$\checkSizeChange(\au)$}
    \label{protocol:check-size-change}
    \begin{algorithmic}[100]
        \If{$\au.\phase = \nPhase$}
        % \If{$\au.\missingV > 2.5 \cdot \au.\est$\text{ or }$\au.\missingV < \au.\est/4$}
        \If{$\au.\missingV \not\in \{0.25 \cdot \au.\est, \ldots, 2.5 \cdot \au.\est\}$}
            \State {$\au.\phase \gets \wPhase$}   
            \State {$\au.\timer \gets 1$}
            \Comment{Waiting for the other agents to detect the size change}
        \EndIf
        \EndIf
    \end{algorithmic}
\end{algorithm}

Eventually all agents will notice the large discrepancy between $\missingV$ and $\est$, and move to the $\wPhase$. The $\wPhase$ is followed by the $\cPhase$ (explained in the $\timerRoutine$). In the $\cPhase$ all agents generate one geometric random variable (stored in $\RG$) and start propagating the maximum value. 
For the purpose of representation, we assume the agents generate a geometric random variable in one line (line 4 in $\ref{protocol:timer-routine}$).\footnote{Alternatively, the agents could generate a geometric random variable through $O(\log n)$ consecutive interactions, each selecting a random coin flip (\textbf{H} or \textbf{T}). In this alternative version, we should make the $\wPhase$ longer.}

Once the $\cPhase$ is finished all agents will update their $\est$ to the maximum geometric random variable they have seen and switch to the $\nPhase$ again.

% \todoi{reviewer: Explain the creation of independent geometric random variables that the agents use in the updating phase. Define GRV in the pseudocode of Protocol 6.}
\begin{algorithm}[ht]
    \floatname{Algorithm}{Protocol}
    \caption{$\timerRoutine(\au)$}
    \label{protocol:timer-routine}
    \begin{algorithmic}[100]
        \State {$\au.\timer \gets \au.\timer + 1$}
        \If{$\au.\timer > 12 \cdot \au.\missingV$}
            \If{$\au.\phase = \wPhase$}
                \State {$\au.\RG \gets$\text{a new geometric random variable}}
                \State {$\au.\phase \gets \cPhase$}   
                \State {$\au.\timer \gets 1$}
            \EndIf
            
            \If{$\au.\phase = \cPhase$}
                \State {$\au.\est \gets \au.\RG$}
                \State {$\au.\phase \gets \nPhase$}  
                \State {$\au.\timer \gets 0$}
            \EndIf
        \EndIf
    \end{algorithmic}
\end{algorithm}

Recall that the agents remain in the $\nPhase$ as long as their $\missingV$ and $\est$ are fairly close. They continue changing their $\group$ values and send group signals as described earlier.

\begin{algorithm}[ht]
%   \floatname{Algorithm}{Protocol}
    \caption{$\EpidemicMaxRG(\agent\ \au, \agent\ \av)$}
    \label{protocol:epidemic-g}
    \begin{algorithmic}[100]
        \If{$\au.\phase = \av.\phase \And \au.\phase \neq \wPhase$}
            \State{$\au.\RG , \av.\RG \gets \mathrm{max}(\au.\RG , \av.\RG)$ }
        \EndIf
    \end{algorithmic}
\end{algorithm}

Intuitively, for each $\group$ value, about $n/2^i$ agents will hold $\group = i$, and try to boost $\signalArray[i]$ by setting it to the max $= \Theta(i)$. As the value of $i$ grows, the number of agents with $\group = i$ decreases, but their signals get stronger since the agents enhance a $\group$ signal $i$ proportional to $i$. In a normal run of the protocol, the agents expect to have positive values in $\signalArray[i]$  for $\group$ values between $[\log \ln n, \log n]$. 
% However, if the adversary adds or removes agents the set of indices with positive values in the $\signalArray$ must change and the agent will re-calculate their size estimation according to the change. 

\begin{toappendix}
\section{Simulation results}
In this section, we present our simulation results for \cref{protocol:dynamic-counting-all}.
We present two separate simulation results, one starting with a uniformly random initial configuration (\cref{fig:groups-signals}) in which each agent starts with a random number (bounded by $60$) in each of their $\group$ and $\est$ fields. Additionally, we set a random integer of $\Theta(i)$ in each index of their $\signalArray$. 
In \cref{fig:groups-signals}, we depict the $\group$ and $\signalArray$ fields of the agents in a population size $n = 10^6$. We can observe the changes in the $\group$ and $\signalArray$ within multiple snapshots from the population. 

% \todo{DD: This section should reference the figures in it.}

\begin{figure}[!htbp]
     \centering
     \begin{tabular}[c]{cc}
     \begin{subfigure}[c]{0.45\textwidth}
        %  \centering
         \includegraphics[trim=48 0 65 0,clip,width=\textwidth]
         {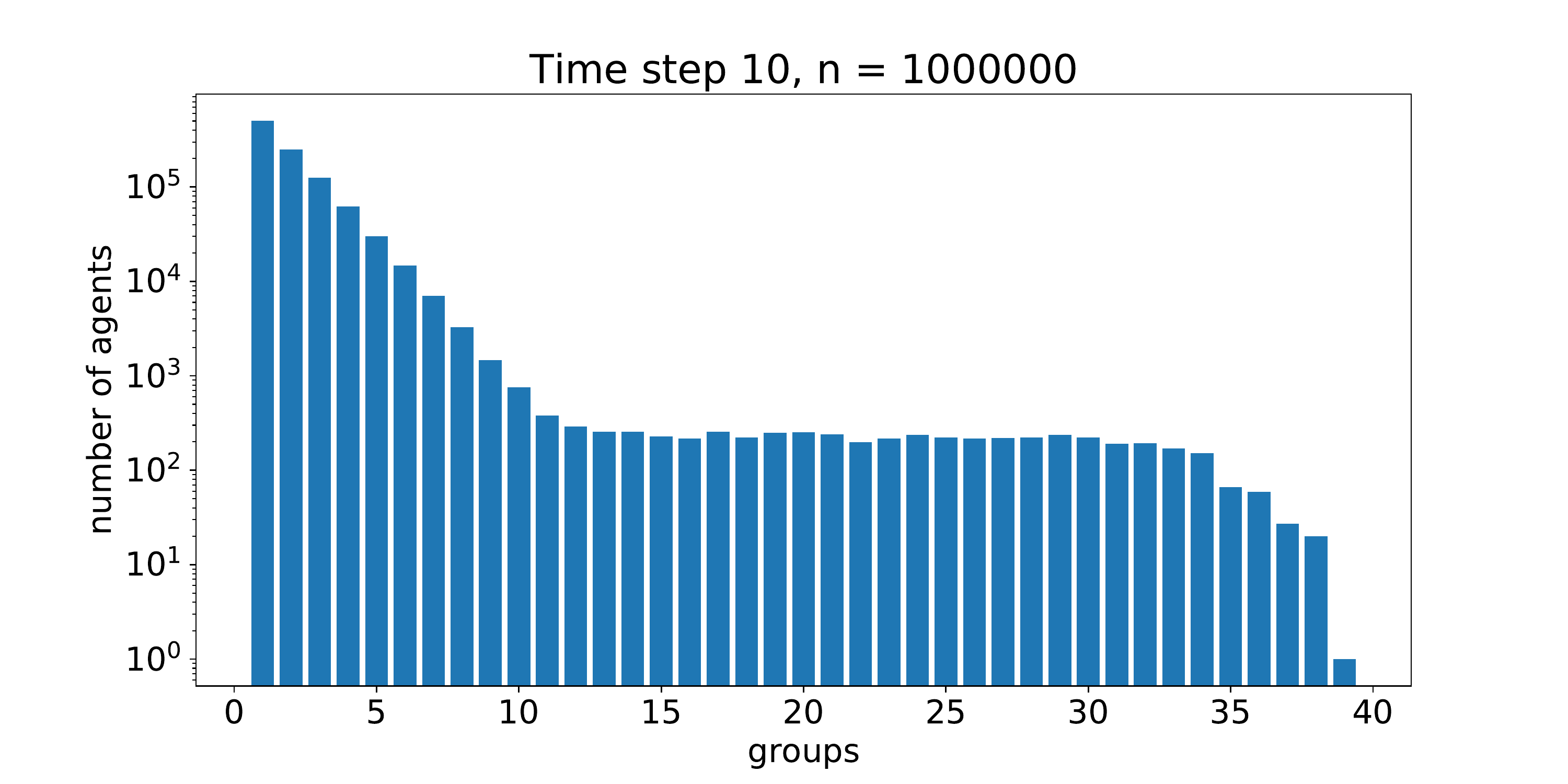}
         \caption{\footnotesize
          The distribution of $\group$ values at $t=10$. Observe that for all $\group$ values less than $10$, the distribution is close to the stationary distribution in which about $\frac{n}{2^i}$ agents have $\group = i$. However, it is too soon for all ``unwanted'' group values to disappear.}
         \label{fig:groups1}
     \end{subfigure} &
     \begin{subfigure}[c]{0.45\textwidth}
        %  \centering
         \includegraphics[trim=48 0 65 0,clip,width=\textwidth]
         {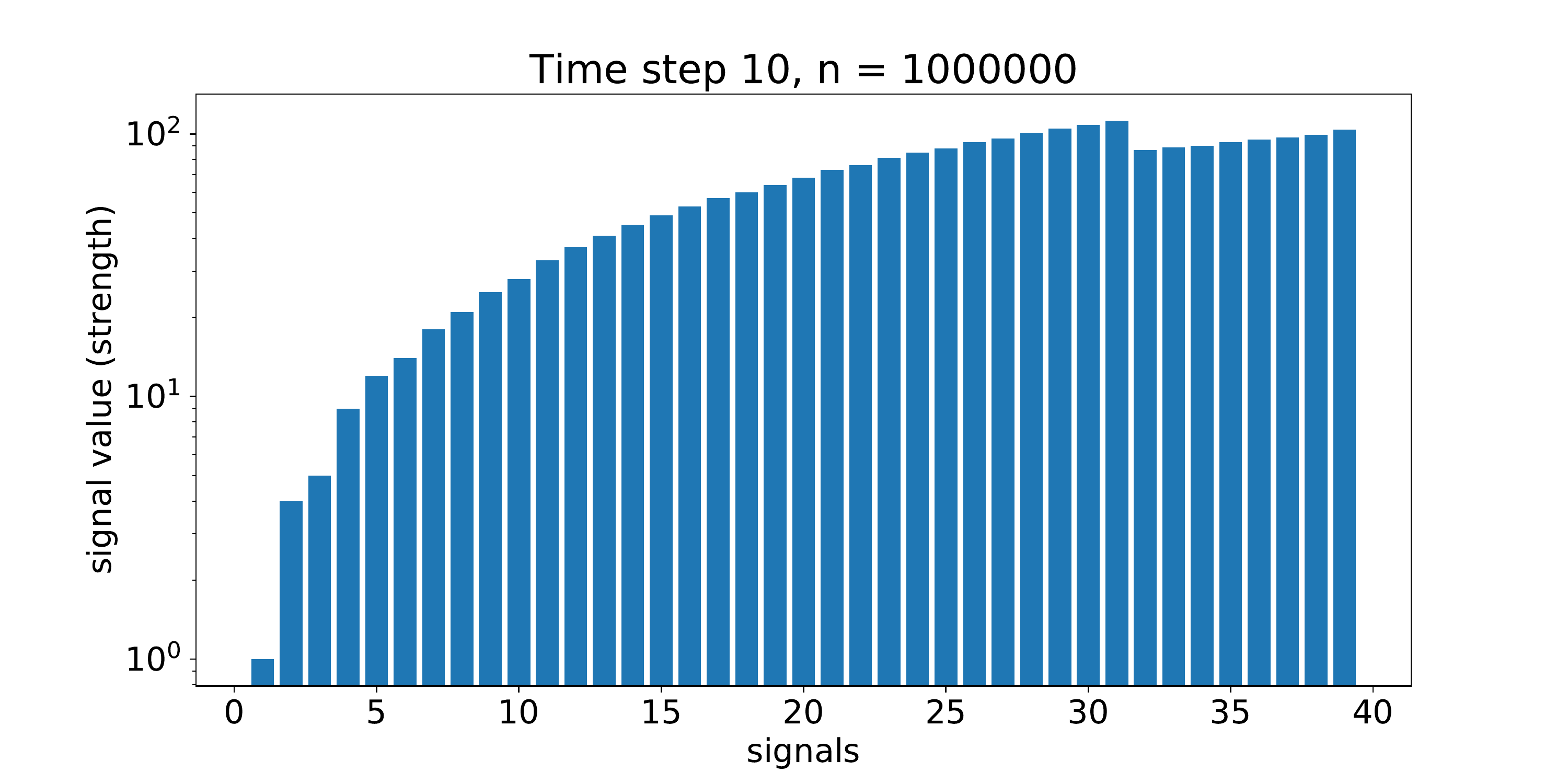}
         \caption{\footnotesize
          Snapshot of the population at time $10$: the distribution of $\signalArray$ values. The data shows $min(u.\signalArray[i]) $ for every index $i$ and agent $u \in \calA$. Note that, for every index $i$, $\signalArray[i]$ now has a positive value.}
         \label{fig:signals1}
     \end{subfigure}\\
     \hfill
     \begin{subfigure}[c]{0.45\textwidth}
        %  \centering
         \includegraphics[trim=48 0 65 0,clip,width=\textwidth]
         {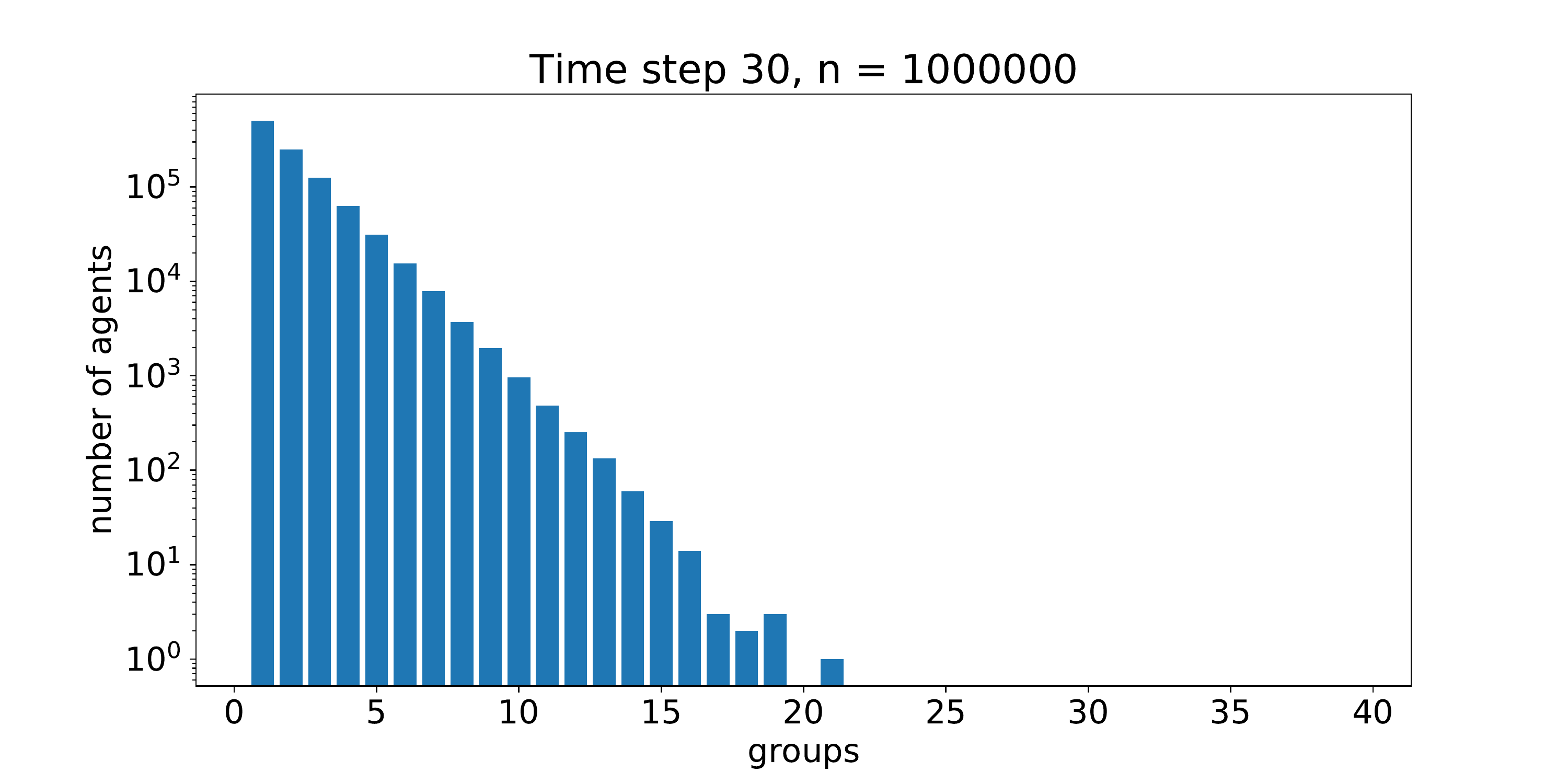}
         \caption{\footnotesize
          Snapshot of the population at time $30$: the distribution of $\group$ values reached the stationary distribution. }
         \label{fig:groups2}
     \end{subfigure}&
     \begin{subfigure}[c]{0.45\textwidth}
        %  \centering
         \includegraphics[trim=48 0 65 0,clip,width=\textwidth]
         {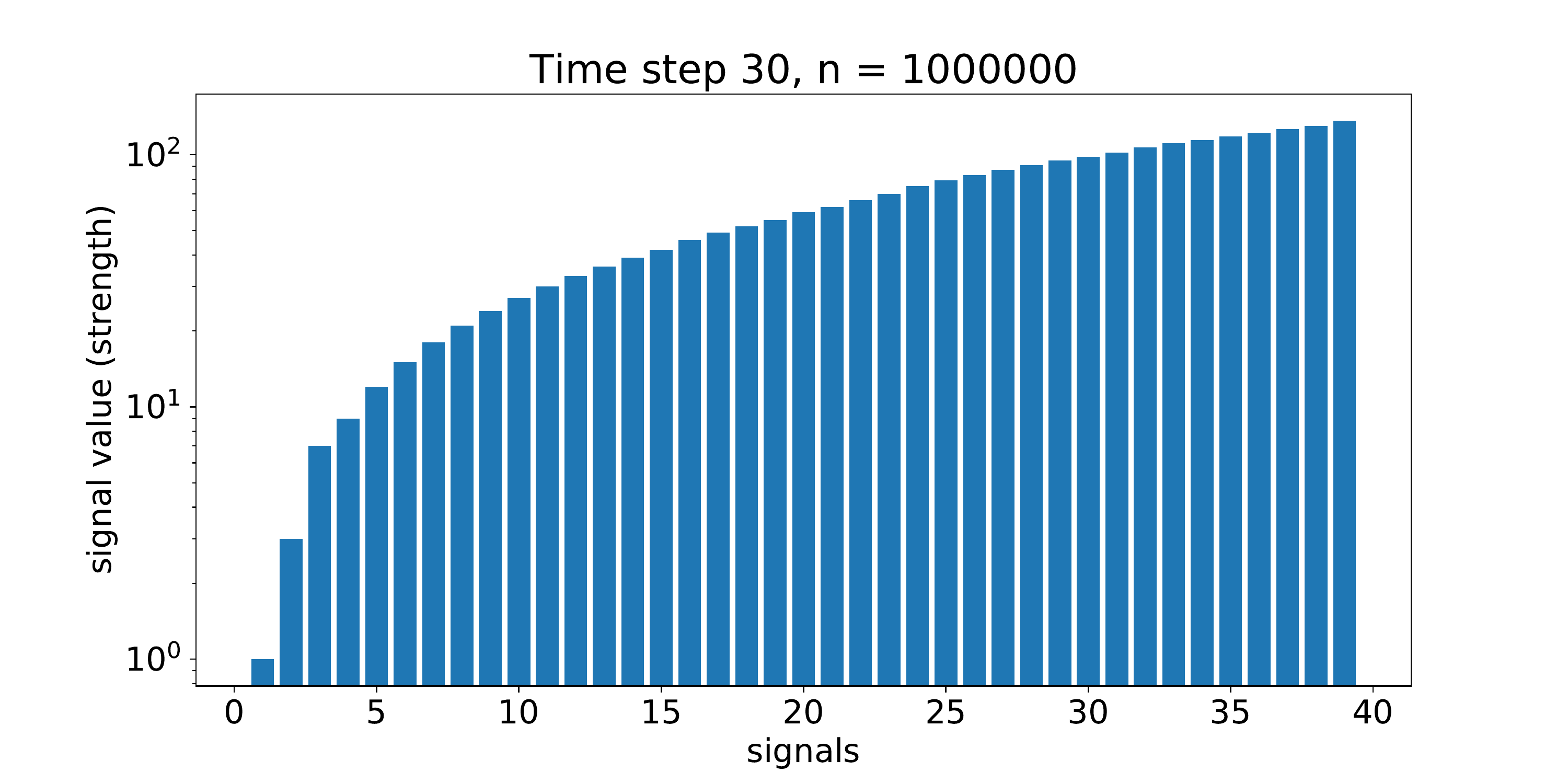}
         \caption{\footnotesize
         Snapshot of the population at time $30$: even though there is no agent with $\group > 25$, the $\signalArray$ for all $\group$s are still positive. }
         \label{fig:signals2}
     \end{subfigure}\\
     \hfill
     \begin{subfigure}[c]{0.45\textwidth}
        %  \centering
         \includegraphics[trim=48 0 65 0,clip,width=\textwidth]
         {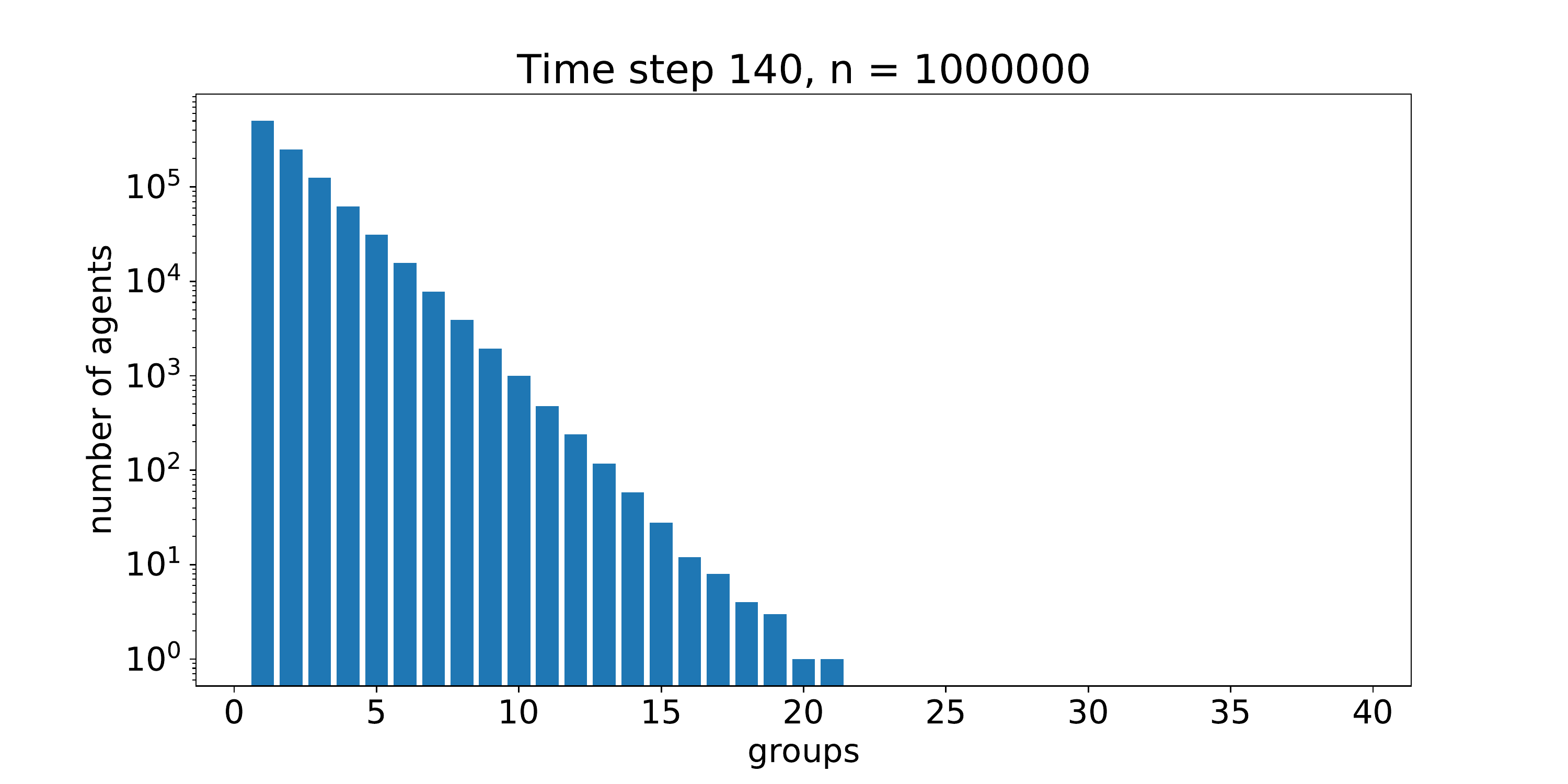}
         \caption{\footnotesize
          Snapshot of the population at time $140$: the distribution of $\group$ values reached the stationary distribution and stays the same for polynomial time. }
         \label{fig:groups3}
     \end{subfigure}&
     \begin{subfigure}[c]{0.45\textwidth}
         \centering
         \includegraphics[trim=48 0 65 0,clip,width=\textwidth]
         {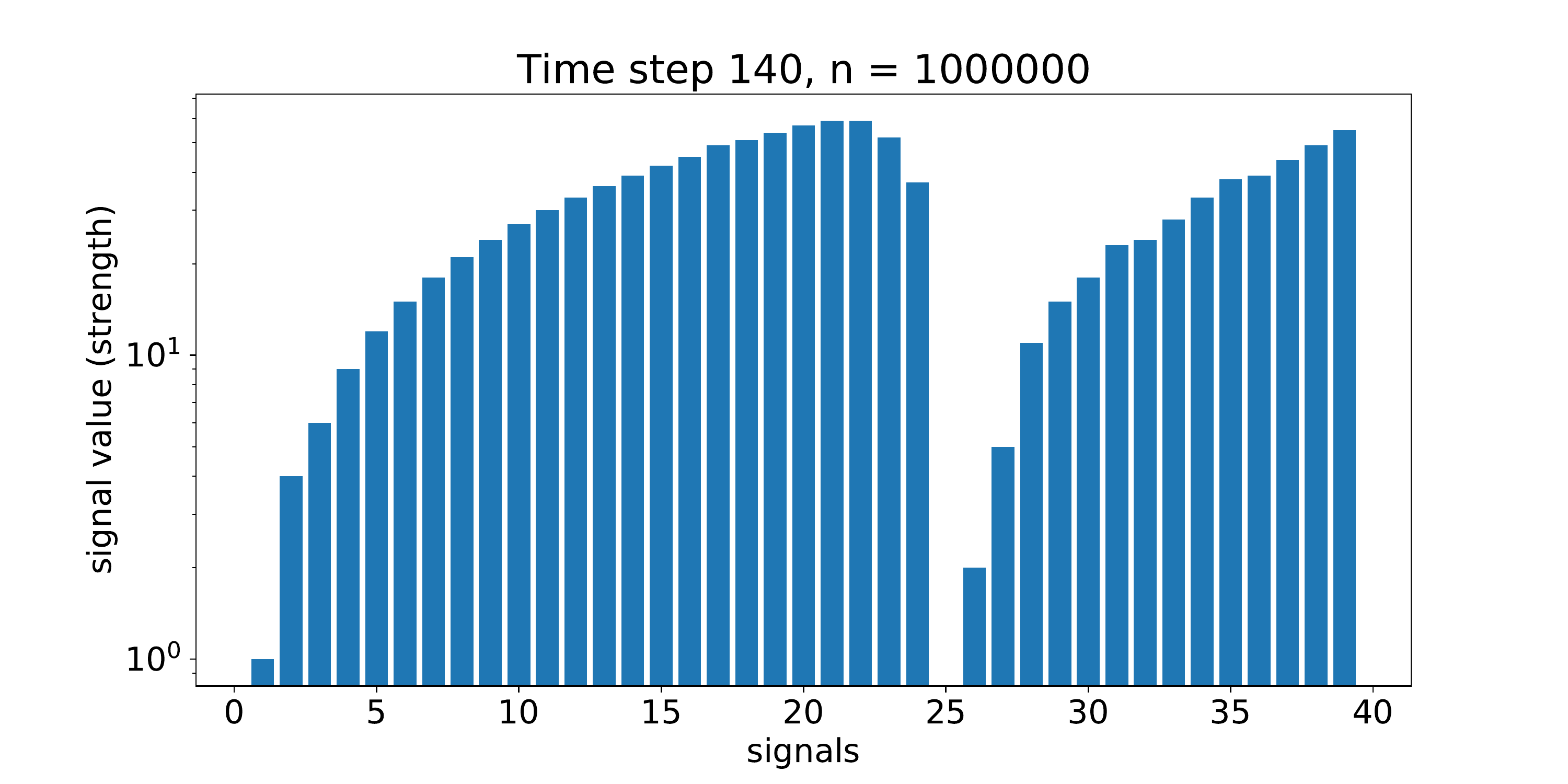}
         \caption{\footnotesize
         Snapshot of the population at time $140$: the emergence of the $\missingV = 25$ which is in $O(\log n)$ for the entire population. Note that in our protocol the agents don't wait for all values of $25$ and greater to disappear. Having $\signalArray[i] = 0 $ for one index $i$ suffices to re-calculate the $\est$.}
         \label{fig:signals3}
     \end{subfigure}\\
     \end{tabular}
     
     \caption{
     \todoi{DD: The text in the figures should be made larger.}
     Simulation results for population size =$10^6$.
     Initializing each agent with a random $1 \leq \group < 30$, and for every $1 \leq i \leq 60$, $0 \leq \signalArray[i] \leq 3 \cdot i + 1$ in $\nPhase, \wPhase,$ or $\cPhase$ with probability $1/3$. 
     Plots of the $\signalArray$ for a randomly initialized population. 
     The x-axis shows all the indices in the $\signalArray$ of the agents (bounded by $40$ in the simulation). On the y-axis, and every index $i$, we take the minimum pairwise value of $u.\signalArray[i]$ for all $u \in \calA$. 
     }
     \label{fig:groups-signals}
\end{figure}

\begin{figure}[!htbp]
     \centering
     \begin{subfigure}[b]{\textwidth}
         \centering
         \includegraphics[trim=150 1 150 1,clip,width=\textwidth]
         {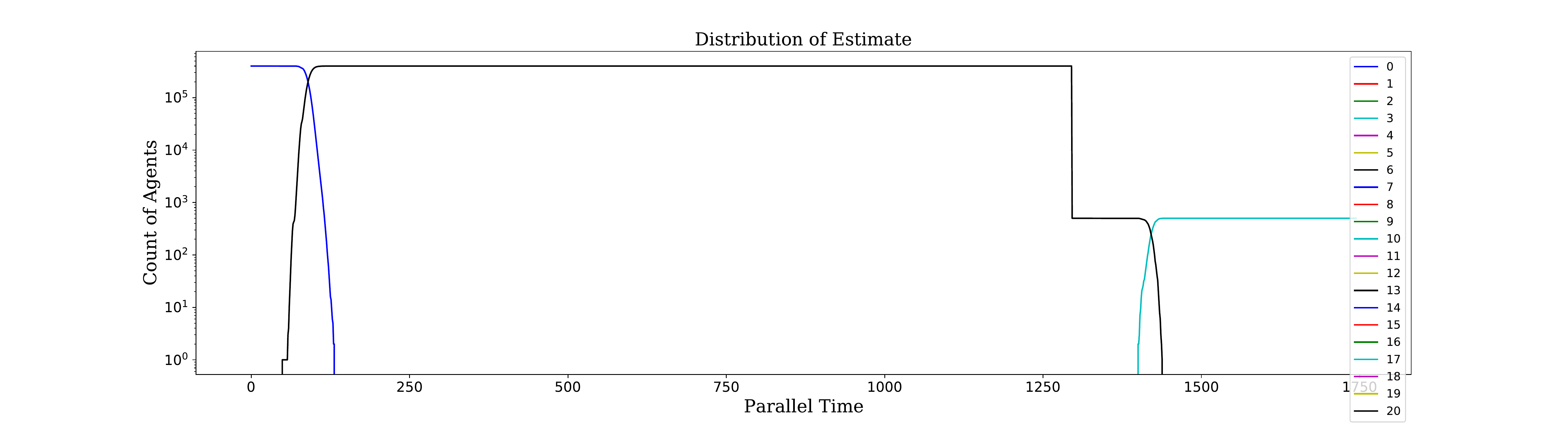}
         \caption{\footnotesize
          The distribution of $\est$ during the course of a computation. First, the agents agree on $\est = 20$ and keep their estimate throughout the computation. This simulation shows how the agents update their $\est$ once we removed $O(n)$ agents from the original population.}
         \label{fig:est}
     \end{subfigure} 
     \hfill
     
     \begin{subfigure}[b]{\textwidth}
         \centering
         \includegraphics[trim=150 1 150 1,clip,width=\textwidth]
         {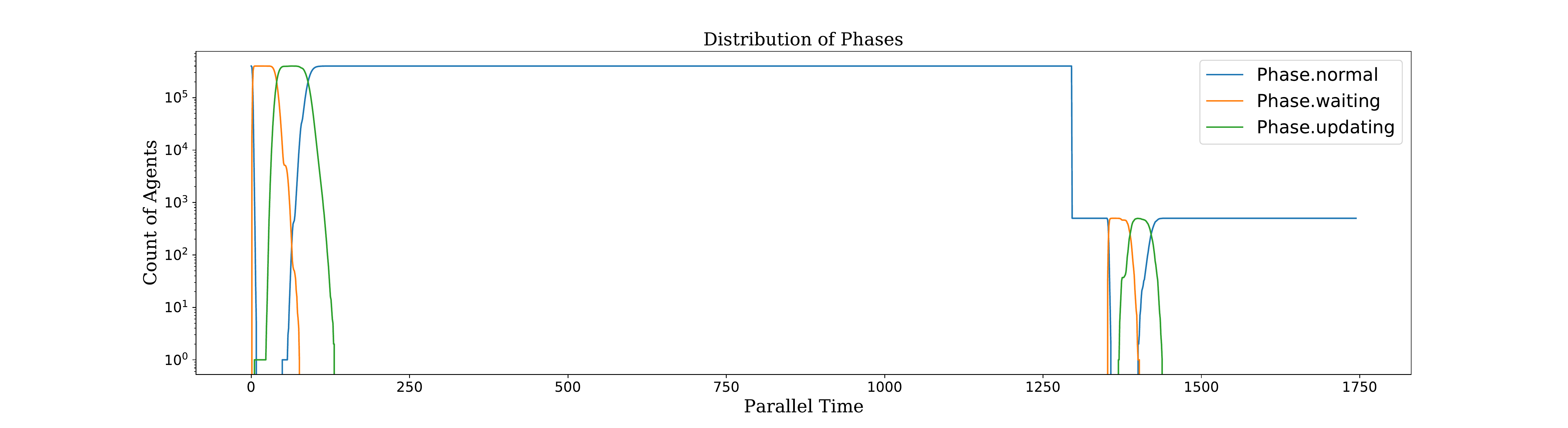}
         \caption{\footnotesize
          The distribution of $\phase$ values. This simulation shows how the agents update their $\est$ by going through $\wPhase$ and $\cPhase$ consecutively.}
         \label{fig:phase}
     \end{subfigure}
     \caption{
     Simulation results for population size $ n = 400000 \approx 2^{18}$.
     Initializing each agent with $ \group = 1 $, and an empty $\signalArray$ in $\nPhase$. First, the agents calculate the $\est$ for $n=2^{18}$, then at time $~1350$ we remove all but $500$ agents randomly which results in updating the agents' $\est$ from $20$ to $10$. 
     }
     \label{fig:phase-est}
\end{figure}

In our second simulation, we initialized the population with default values (starting in the initialized setting), however, after the population converges to an $\est$ of $O(\log n)$ we remove agents uniformly at random to simulate an adversarial initialized population. The result of this simulation is shown in \cref{fig:phase-est}.
\end{toappendix}

\section{Analysis of $\dcountingA$}

\begin{toappendix}
\subsection{Useful time bounds}\label{subsec:time-analysis}
Let $n$ and  $\calA$ be the population size and the set of agents in the population respectively. Let $I(t, u)$ represent the number of interactions involving agent $u$ by the time $t$. Also, let $G_{u,t}$ represent the $\group$ value for agent $u$ at time $t$ calculated via the rules defined in~\cref{protocol:stationary-geometric}.

\begin{lemmarep}\label{lem:fixed-interval-interactions-bound}
    For any $d \geq 3$ and $2d/3 \geq c \geq 1$ during $d n \ln n $ interaction, all agents have at least $2(d-\sqrt{dc}) \ln n$  and at most $2(d+\sqrt{3dc/2}) \ln n$ interactions with arbitrary large probability $1-n^{1-c}$. 
\end{lemmarep}

\begin{appendixproof}
    Let us consider a fixed agent $u$. Recall that we defined $I(t, u)$ to be the number of interactions involving agent $u$ in $t$ time. Note that $I(t, u)$ has a binomial distribution with parameters $\B(dn \ln n, 2/n)$ with $\mu = 2d \ln n$ for $t=d\ln n$. We can get a tight bound on the expected value using a straightforward Chernoff bound:
    
    For the lower bound:
    \begin{align}
        \Pr{I(t, u) \leq (1-\delta) \mu} 
        & \leq 
        e^{-\delta^2\mu / 2} \\
        \Pr{I(t, u) \leq (1-\delta) 2d \ln n} 
        & \leq 
        e^{-\delta^2(2d \ln n) / 2} \\
        \Pr{I(t, u) \leq (1-\delta)2d \ln n} 
        & \leq n^{-\delta^2 d}\\
        \Pr{I(t, u) \leq 2(d - \sqrt{dc}) \ln n} 
        & \leq 
        n^{-c} && \text{setting $\delta = \sqrt{c/d}\leq 1$}
    \end{align}
    
    With a union bound we can show that $\Pr{I(t, u) \leq 2(d - \sqrt{dc}) \ln n} \leq n^{1-c}$ for all agents. 
    
    and for the upper bound:
    \begin{align}
        \Pr{I\qty(t, u) \geq \qty(1+\delta) \mu} 
        & \leq 
        e^{-\delta^2\mu / (2+\delta)} \\
        \Pr{I\qty(t, u) \geq \qty(1+\delta) 2d \ln n} 
        & \leq 
        e^{-\delta^2\qty(2d \ln n) / 3} && \text{for $0 \leq \delta \leq 1 $}\\
        \Pr{I\qty(t, u) \geq \qty(1+\delta) 2d \ln n} 
        & \leq n^{-2\delta^2 d / 3}\\
        \Pr{I\qty(t, u) \geq 2\qty(d+\sqrt{3dc/2}) \ln n} 
        & \leq n^{-c} && \text{setting $\delta = \sqrt{3c/2d}\leq 1 $}
    \end{align}
    
    A union bound shows that $\Pr{I\qty(t, u) \geq  2\qty(d+\sqrt{3dc/2}) \ln n} \leq n^{1-c}$ for all agents. 
\end{appendixproof}

\begin{corollary}\label{cor:fixed-interval-interactions-bound}
    For any $d \geq 3$, during $d n \ln n $ interaction, all agents have at least $0.2d \ln n$  and at most $4d \ln n$ interactions with probability $1-n^{1 - 2d / 3}$. 
\end{corollary}
\begin{proof}
    Setting $c = 2d/3$ in \cref{lem:fixed-interval-interactions-bound} results in the above bounds.
\end{proof}

\end{toappendix}
% --------------------------------------------------------------------------------------
% DYNAMIC GROUPS
% --------------------------------------------------------------------------------------
\subsection{Bound on the group values}\label{subsec:group-bound-analysis}

Recall that the agents calculate a dynamic $\group$ value by following the rules of~\cref{protocol:stationary-geometric}. As described in this protocol, the agents move through different $\group$ values according to the Markov chain shown in~\cref{fig:sg-chain}.\footnote{The truncated chain mapping all states $k+1$, $k+2$, $\ldots$ to $k+1$ is also known as the ``winning streak''~\cite{lovasz1998reversal}.}

\begin{figure}[ht!]
\centering
    \includegraphics[width=\textwidth, draft=false]{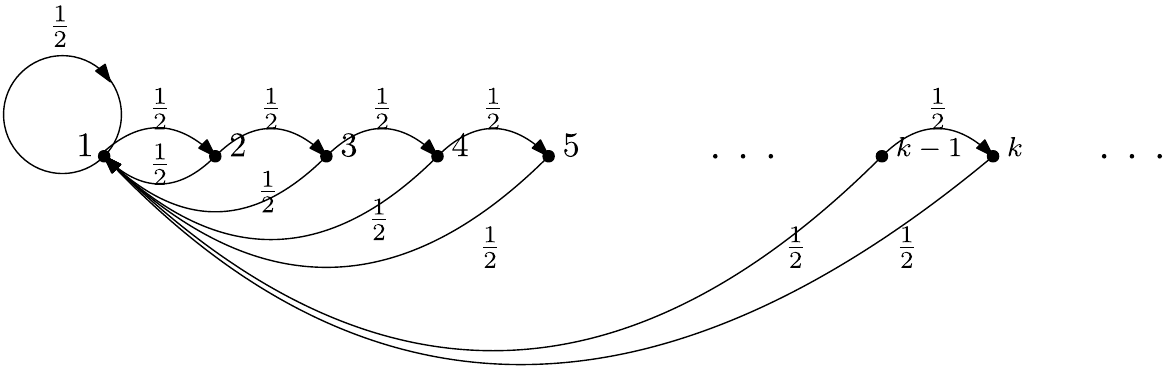}
    \caption{The infinite chain of $\group$ values. }
	\label{fig:sg-chain}
\end{figure}

In this part, we analyze the distribution of $\group$ values. Note that at the very beginning of the protocol, the $\group$ values are rather chaotic since the agents might start holding any arbitrary $\group$ values that are much larger than $\log n$. However, after all agents reset back to $\group =1$, we can show for each $\group =k$, $\Pr{\group = k} \approx \frac{1}{2^k}$.

In the rest of this section we assume the initialized setting for simplicity. Later on we show how we can generalize our results to any arbitrary initial configuration. Recall that $G_{u,t}$ stands for the $\group$ value of agent $u$ at time $t$ and $I(t, u)$ shows the number of interactions involving this agent by the time $t$. Note that with this definition, $G_{u,t}$ is equal to $k$ (for $k < I(t, u)$) if and only if agent $u$ generates the sequence of $[HTTT\ldots T]$ (\textbf{H} followed by $k-1$ \textbf{T}s) during its last $k$ interactions. Thus, we have: 
\begin{equation}\label{eq:group-prob}
    \forall k\in \N, \quad  1 \leq k < I(t, u): \Pr{G_{u,t} = k} = \frac{1}{2^k} 
\end{equation}
With this definition $G_{u, t}$ is undefined for any agents that has not generated $H$ yet. 
In other words, the values $G_{u,t}$ are ``close to geometric'' in the sense that they are independent and have probability equal to a geometric random variable on all values $k < I(t,u)$.

\begin{observation}\label{obs:group-val-independence} 
    For agents $u_1, u_2, \ldots, u_n$, and the values $k_i < I(t, u_i)$, for $1\leq i \leq n$: 
    \begin{equation*}
        \Pr{G_{u_1, t} = k_1 , G_{u_2, t}=k_2, \ldots, G_{u_n, t}=k_n } = \prod_{i=1}^{n}\Pr{G_{u_i, t} = k_i}  
    \end{equation*}
\end{observation}

Next we bound the maximum $\group$ value that has been generated by any agent. Let $M_t = \max_{u \in \calA} G_{u, t}$ be the maximum value of $G_{u, t}$ across the population at time $t$. 

% ME: I couldn't simply use the results of the termination paper since every single lemma
% in that paper is tighten for probability 1-1/n (not arbitrary large probability)
\begin{lemmarep}\label{lem:max-min-group-values}
    Let $c \geq 2$ and let $t$ be a time such that all agents have at least $c \log n$ interactions. In a population of size $n$, $\frac{1}{d}\log n \leq M_t $ with probability at least $1-\exp{-n^{1 - 1/d}}$ and $M_t < c\log n $ with probability at least $1-n^{1-c}$.
\end{lemmarep}
% \newtheoremrep{Lemma}{lem}
%  this is a test.

\begin{appendixproof}
    Recall that $G_{u, t}$ is the $\group$ value of agent $u$, $I( t, u)$ is the number of interactions that this agent had by the time $t$, and $M_t$ is the maximum value of $G_{u, t}$ for all $u \in \calA$ at time $t$. 
    Since all agents have had at least $c \log n$ interactions, for all values of $k < c\log n$, $\Pr{G_{u, t} = k} = \frac{1}{2^k}$.
    To have $M_t \geq c\log n$, at least one agent must have generated a $\group$ value greater than or equal to $c \log n$:
    
    \begin{align*}
        \Pr{M_t \geq c\log n} 
        & = 
        \Pr{\qty(\exists u \in \calA)G_{u, t} \geq c \log n} \\
        & \leq 
        n \cdot \qty(\frac{1}{2})^{ c\log n} \quad \text{by the union bound}\\
        & = 
        n^{1-c}
    \end{align*}
    
    For the other direction, to have $M_t> \frac{1}{d} \log n $, at least one agent must have generated $\group > \frac{1}{d}\log n$. By \cref{obs:group-val-independence}:
    
    \begin{align*}
        \Pr{M_t < \frac{1}{d} \log n} 
        & = 
        \Pr{\qty(\forall u \in \calA)G_{u, t} < \frac{1}{d} \log n} \\
        & = 
        \qty(1-\frac{1}{2^{1/d\log n}} )^n \\
        & = 
        \qty(1-\frac{1}{n^{1/d}} )^n\\
        & = 
        \qty(1-\frac{n^{1-1/d}}{n} )^n \\
        & \leq 
        \exp{-n^{1-1/d}}.
        \qedhere
    \end{align*}    
\end{appendixproof}

% \todo{DD: I don't think it's worth a Corollary just to set $d$ to some constant. Just say ``let $d=$'' when you use \cref{lem:max-min-group-values}.}

% \begin{corollary}\label{cor:max-min-group-values}
%     Let $c \geq 2$ and let $t$ be a time such that all agents have at least $c \log n$ interactions. In a population of size $n$, $0.9 \log n \leq M_t $ with probability at least $1-\qty(\exp{-n^{0.1}} )$ and $ M_t < c\log n $ with probability at least $1-\qty(n^{1-c})$.
% \end{corollary}

Additional to bounds on the maximum $\group$ value, we need to calculate the bounds for the first missing $\group$. Note that the maximum $\group$ value has a large variance, however, we can prove a tight bound for the first missing $\group$ since to have $\missingV = k$, for all values $i$ that are less than $k$, $\exists u \in \calA$ such that $u.\group = i$ .

\subsection{Bounds on the first missing group}
In this part, we analyze the bounds for the first $\group$ value that has no support, i.e., the value  
$\min \{k \in \N^+ \mid (\forall u \in \calA)\ u.\group \neq k \}$.
Considering $n$ i.i.d. geometric random variables,
the \emph{first missing value} to be the smallest integer not appearing among the random variables.
% that does not appear as the ``first missing value''. 
% Define the \emph{first missing value} among $n$ i.i.d. geometric random variables $X_1,\ldots,X_n$ to be $\missingV = \min \{k \in \N \mid k \not\in \{X_1,\ldots,X_n\} \}$.
The first missing value has been studied in the literature~\cite{louchard2005number,prodinger2021philippe, louchard2004moments} as the ``the first empty urn'' (see also ``probabilistic counting'' \cite{flajolet1985probabilistic}) but for simplicity we use a loose bound for our analysis. 
% the first empty urn in an urn model where we throw n GEOM(1/2) RV

% In our protocol, the agents generate $\group$ values rather than actual geometric random variable. However, in~\cref{subsec:group-analysis}, we showed that for any agent $u$, $\Pr{u.\group = i} = \frac{1}{2^i}$.

% snapshot
\begin{lemmarep}\label{lem:first-missing-group}
    Let $\delta > 0$, $0 < \epsilon < 1$ and let $t$ be a time such that all agents have at least $(1+\delta) \log n$ interactions. 
    Define $\missingV_t = \min \{k \in \N \mid (\forall u \in \calA)\ u.\group \neq k\}$ at time $t$. 
    Then, $\missingV_t > (1-\epsilon) \log n$ with 
    probability at least $(1-\epsilon) \log (n) \cdot \exp{-n^{\epsilon}} $ and $\missingV_t \leq (1+\delta)\log n$ with probability at least $1-\qty(\frac{1}{ n^{\delta/2} })^{(2+\delta) \log n}$.
\end{lemmarep}
\begin{appendixproof}
    Recall that we use $I(u,t)$ to represent the number of interactions agent $u$ had by time $t$. 
    Also, we use $G_{u,t}$ to show $u.\group$ at time $t$. 
    
    For any $\group$ value $k$ that $k < I(u,t)$, we have $\Pr{G_{u,t} = k} = \frac{1}{2^k}$ (\cref{eq:group-prob}). We say group $i$ is \emph{missing} if $i \neq G_{u, t}$ for all $u \in \calA$. 
    Since the $G_{u,t}$'s are independent (across different agents $u$ for a fixed $t$),
    For any $1 \leq i \leq c\log n$, the probability that $i$ is missing is $\qty(1-\frac{1}{2^{i}})^n \approx \exp{-\frac{n}{2^i}}$.
    We take a union bound on all $1 \leq i \leq \epsilon \log n$:
    
    \begin{align*}
        \Pr{\missingV_t \leq (1-\epsilon) \log n} 
        & \leq 
        \sum_{i=1}^{(1-\epsilon) \log n} \qty( 1-\frac{1}{2^{i}} )^n \\
        & \leq
        \sum_{i=1}^{(1-\epsilon) \log n} 
        \exp{-\frac{n}{2^i}}
        \\
        % & \leq 
        % l \cdot \qty( 1-\frac{1}{2^{\epsilon \log n}} )^n \\
        & \leq 
        (1-\epsilon) \log(n) \cdot \exp{-\frac{n}{2^{(1-\epsilon) \log n}}} 
        && \text{since $\qty(1 + \frac{x}{n})^n \leq \exp{x}$}\\
        & = 
        (1-\epsilon) \log (n) \cdot \exp{-n^{\epsilon}} 
        % && \text{setting $l = \epsilon \log n $}
    \end{align*}
    
    For the upper bound, 
    \todoi{Add something about negative dependence}
    
    % n=3
    % generate n geometric r.v. G1, G2
    % want Pr[1 and 2 both get generated]
    % Pr[some Gi = 1] = 1 - Pr[all Gi != 1] = 1 - (1/2)^3 = 7/8
    % Pr[some Gi = 2] = 1 - Pr[all Gi != 2] = 1 - (1-1/4)^3 = 1- (3/4)^3 = 1- 27/64 = 37/64 ~ 0.58
    
    % Pr[some Gi = 3] = 1 - Pr[all Gi != 3] = 1 - (1-1/8)^2 = 1- (7/8)^3 = 1- 343/512 
    %
    % n=2
    % assuming independence: Pr[some Gi=1 and some other Gi=2] = (3/4)*(7/16) = 21/64
    % actual Pr[some Gi=1 and some other Gi=2] = 2* (1/2 * 1/4 ) = 1/4 = 16/64
    % actual Pr[some Gi=1 and some other Gi=2] =  n!  * (1/2 * 1/4 ) = 1/4 = 16/64
    
    % n=3
    % assuming independence: Pr[some Gi=1 and some other Gi=2] = (7/8)*(37/64) = 259/512 ~ 0.5
    % actual Pr[some Gi=1 and some other Gi=2] =  
    % 1/8 * [(1/2 * 1/4) + (1/4 * 1/2) +(1/2 * 1/4) +(1/4 * 1/2) + (1/2 * 1/4) +(1/4 * 1/2)] = 6 * 1/8 * 1/8
    % + 1/2 * [(1/2 * 1/4) + (1/4 * 1/2) +(1/2 * 1/4)] = 1/2 * 3/8
    % + 1/4 * [(1/2 * 1/4) + (1/4 * 1/2) +(1/2 * 1/4)] = 1/4 * 3/8
    %= 6 * 1/8 * 1/8+  3/16 + 3/32] = 3/32 + 6/32 + 3/32= 12/32 ~ 0.37
    %
    % 1/8 * [(1/2 * 1/4) + (1/4 * 1/2) +(1/2 * 1/4) +(1/4 * 1/2) + (1/2 * 1/4) +(1/4 * 1/2)] = 6 * 1/8 * 1/8 = 6/64 (Pr[exactly 1 agent generates >= 3])
    % plus
    % Pr[112 or 121 or 211 or 221 or 212 or 122] = 3*(1/2 * 1/2 * 1/4) + 3*(1/2 * 1/4 * 1/4) = 3/16 + 3/32 = 9/32
    % sum = 6/64 + 9/32 = 24/64 ~ 0.38
    
    %= 6 * 1/8 * 1/8        

    %%%%%%%%%%%%%%%%%%%%
    %%     Proof 1    %%
    %%%%%%%%%%%%%%%%%%%%
    \begin{align*}
        \Pr{\missingV_t > U} 
        & \leq
        \prod_{i=1}^{U} \Pr{\qty(\exists u \in \calA)G_{u, t} = i}\\
        & \leq
        \prod_{i=1}^{U} \qty( n \cdot \qty( \frac{1}{2^{i}} ) ) && \text{By the union bound}
        \\ & =
        n^U \cdot \prod_{i=1}^{U} \frac{1}{2^{i}}
        \\ & =
        n^U \cdot  \frac{1}{2^{\sum_{i=1}^U i}}
        \\ & =
        \frac{n^U}{2^{U(U+1)/2}}
        \\ & =
        \qty( \frac{n}{2^{(U+1)/2}} )^U
    \end{align*}
    
    Substituting $U = (2+\delta) \log n$,
    
    \begin{align*}
        \Pr{\missingV_t > (2+\delta) \log n} 
        & \leq
        \qty( 
            \frac{n}
            { 2^{[1 + (2+\delta) \log n]/2} }
        )^{(2+\delta) \log n}
        \\ & = 
        \qty( 
            \frac{n}
            { \sqrt{2} \cdot 2^{(1 + \delta/2) \log n} }
        )^{(2+\delta) \log n}
        \\ & = 
        \qty( 
            \frac{n}
            { \sqrt{2} \cdot n^{1 + \delta/2} }
        )^{(2+\delta) \log n}
        \\ & = 
        \qty( 
            \frac{1}
            { \sqrt{2} \cdot n^{\delta/2} }
        )^{(2+\delta) \log n}
        \\ & < 
        \qty( 
            \frac{1}
            { n^{\delta/2} }
        )^{(2+\delta) \log n}
        % \\ & <
        % \frac{1}
        % { n^{\delta \log n} }
    \end{align*}

\end{appendixproof}

\begin{toappendix}
    \begin{corollary}\label{cor:first-missing-group}
        Let $c \geq 2 $ and let $t$ be a time such that all agents have at least $c \log n$ interactions. Define $\missingV_t = \min \{k \in \N \mid (\forall u \in \calA)\ u.\group \neq k\}$ at time $t$. Then, $\missingV_t > 0.9 \log n$ with 
        probability at least $1- 0.9 \log n \cdot \exp{-n^{0.1}}$ and $\missingV_t \leq 3\log n$ with probability at least $1- n^{\frac{3 \log n}{2}}$.
    \end{corollary}
    \begin{proof}
        Set $\epsilon = 0.1$ and $\delta = 1$ in~\cref{lem:first-missing-group}.
    \end{proof}
\end{toappendix}

\subsection{Distribution of the group values}\label{subsec:group-dist-analysis}

\begin{figure}[ht!]
\centering
    \includegraphics[width=0.9\textwidth, draft=false]{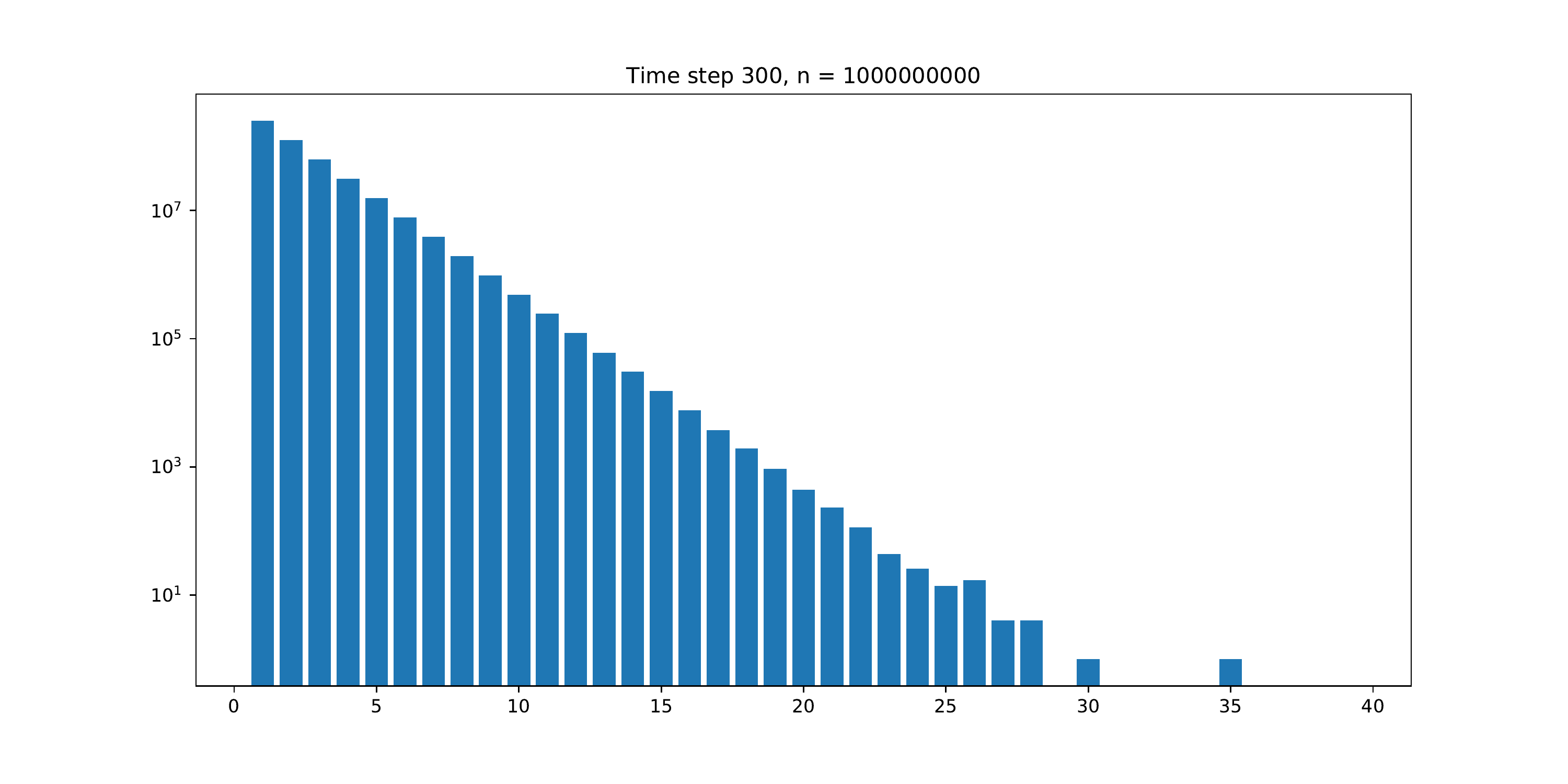}
    \caption{Showing the distribution of $\group$ values after $300$ parallel-time in a population of size $n = 10^9$. The x-axis indicates the different $\group$ values while the y-axis indicates the number of agents in each group. Note that, we are using log-scale for the y-axis. In this snapshot of the population, $\missingV = 29$. Even though, the maximum $\group$ value is $35$ and is much larger than $\missingV$. }
	\label{fig:group-dist}
\end{figure}

So far we proved bounds on the existing $\group$ values. However, in general, we need to show that at a given time $t = \Omega(\log n)$, there are about $\frac{n}{2^k}$ agents having $\group = k$ WHP. The following lemma gives us a lower and upper bound for number of agents in each $\group$:

\begin{lemmarep}\label{lem:group-distribution}
    Let $c \geq 2$, $0 < \epsilon < 1 $, $0 \leq \delta \leq 1$, and let $t$ be a time such that all agents have at least $c \log n$ interactions.
    Let $1 \leq k \leq (1-\epsilon) \log n$, then, the number of agents who hold $\group = k$, is at least $L_k = \qty(1-\delta)\frac{n}{2^k}$ with probability at least $1- \exp{-\frac{\delta^2 \cdot n^{\epsilon}}{2}}$  and at most $U_k = \qty(1+\delta)\frac{n}{2^k}$ with probability at least $1-\exp{-\frac{\delta^2 \cdot n^{\epsilon}}{3}}$.
\end{lemmarep}

\begin{proofsketch}
    The fraction of agents with $\group = k$ is equal to the fraction of heads in of a binomial distribution $\B\qty(n,2^{-k})$ with $\mu = \frac{n}{2^k}$,
    so the Chernoff bound applies. 
\end{proofsketch}

\begin{appendixproof}
    Again, we consider $\group$ values that are less than $I(t, u)$ for all $u \in \calA$.
    Recall that at time $t$, agent $u$ has $\group = k$ with probability $1/2^k$ for $k < I(t, u)$. 
    In \cref{obs:group-val-independence}, we showed the event of having agents holding different $\group$ values are independent, thus, the fraction of agents holding value $\group = k$ is equal to the fraction of heads in of a binomial distribution $\B\qty(n,2^{-k})$ with $\mu = \frac{n}{2^k}$. Thus, we can use Chernoff bound for the tail bounds. For the upper bound and all $0 \leq \delta \leq 1$ we have:
    
    \begin{align*}
        \Pr{\B\qty(n, 2^{-k}) \geq \qty(1+ \delta) \frac{n}{2^k}} 
        &\leq 
        \exp{-\frac{\delta^2 \mu}{3}}\\
        & \leq 
        \exp{-\frac{\delta^2 \cdot n }{3 \cdot 2^k}} \\
        & \leq 
        \exp{-\frac{\delta^2 \cdot n }{3 \cdot 2^{(1-\epsilon) \log n}}} \quad \text{since $k \leq (1-\epsilon) \log n$}\\
        & \leq 
        \exp{-\frac{\delta^2 \cdot n^{\epsilon}}{3}} 
    \end{align*}
    
    Similarly for the lower bound and all $0 \leq \delta \leq 1$ we can write:
    
    \begin{align*}
        \Pr{\B\qty(n, 2^{-k}) \leq \qty(1- \delta) \frac{n}{2^k}} 
        &\leq 
        \exp{-\frac{\delta^2 \mu}{2}}\\
        & \leq  
        \exp{-\frac{\delta^2 \cdot n }{2^{k+1}}} \\
        & \leq  
        \exp{-\frac{\delta^2 \cdot n }{2^{(1-\epsilon) \log n +1}}} \quad \text{since $k \leq (1-\epsilon) \log n$}\\
        & \leq  
        \exp{-\frac{\delta^2 \cdot n^{\epsilon}}{2}}.
        \qedhere
    \end{align*}
\end{appendixproof}

Setting $\epsilon = 0.1$ and $\delta = 1/4$ gives the following corollary.

\begin{corollary}\label{cor:independent-sg-fraction-bounds}
    Let $c \geq 2$, and let $t$ be a time such that all agents have at least $c \log n$ interactions.
    Let $1 \leq k \leq 0.9 \log n$, then, the number of agent who hold $\group = k$, is at least $L_k = \frac{3 \cdot n}{2^{k+2}}$ and at most $U_k = \frac{5 \cdot n}{2^{k+2}}$ with probability at least $1- \exp{-\frac{n^{0.1}}{32}}$ and $1- \exp{-\frac{n^{0.1}}{48}}$ respectively.
\end{corollary}

Let us summarize what we know about the distribution of the $\group$ values and the number of agents holding each $\group$ value at time $t$ in the following theorem. 

% \todo{Here we need a general theorem to talk about maximum $\group$ values and the number of agents holding each $\group$ value. How is it?}
\begin{theoremrep}\label{thm:dynamic-groups}
    Fix a time $t \geq d \ln n$ for $d > 30$, let $M^*_t$ and $\missingV_t$ be the maximum $\group$ value and the $\missingV$ at this time respectively. Then, 
    
    \begin{itemize}
        \item $0.9 \log n \leq M^*_t < 0.1 d\log n $ with probability at least $1-2 \cdot n^{1- d/10} - 2 \cdot n^{1-\frac{2d}{3}}$.
        
        \item $0.9 \log n \leq \missingV_t < 3 \log n $ with probability at least $1 - 4 \cdot n^{1-\frac{2d}{3}}$.
        
        \item The number of agents who hold $\group = k$ for $1 \leq k \leq 0.9 \log n$, is in $\left[\frac{3 \cdot n}{2^{k+2}}, \frac{5 \cdot n}{2^{k+2}}\right]$ with probability at least $1 - 4 \cdot n^{1-\frac{2d}{3}}$.
    \end{itemize}

\end{theoremrep}
\begin{appendixproof}
    By \cref{cor:fixed-interval-interactions-bound}, at time $t \geq d \ln n$ all agents have at least $0.2 d \ln n > 0.1 d \log n$ interactions with probability $1-n^{1-\frac{2d}{3}}$. Conditioning on all agents having $I_U = 0.1 d \log n$ interactions, we can use \cref{lem:max-min-group-values} and \cref{lem:group-distribution} to prove than the maximum $\group$ value is at most $U = 0.1 d \log n$ using the law of total probability:
    
    Let $\calE$ denote the event $(\forall a \in \calA) \ I(t, a) > I_U $, all agents have at least $I_U$ interactions. 
    Let $\calD$ denote an arbitrary event over the population. By the law of total probability we have:
    \begin{align*}
        \Pr{\calD} 
        & = 
        \Pr{\calD | \  \calE} \cdot \Pr{\calE} + \Pr{\calD \ | \ \bar \calE } \cdot \Pr{\bar \calE}\\
        & \geq 
        \Pr{\calD \ | \  \calE}\qty(1-n^{1-\frac{2d}{3}}) + 0 \cdot n^{1-\frac{2d}{3}} 
            && \text{by \cref{cor:fixed-interval-interactions-bound}}
    \end{align*}
    
    Substitute $\calD$ with event $M^*_t < U$. Thus, for the maximum $\group$ value at time $t$ we have: 
    \begin{align*}
        \Pr{M^*_t < 0.1 d \log n} 
        & \geq 
        \qty(1-n^{1-0.1 d})\cdot\qty(1-n^{1-\frac{2d}{3}}) 
            && \text{by \cref{lem:max-min-group-values}}\\
        & \geq 
        1 - n^{1-0.1 d} - n^{1-\frac{2d}{3}} + n^{2 - \frac{23 \cdot d}{30}}\\
        & \geq 
        1-2 \cdot n^{1-0.1 d}
            && \text{since  $n^{1-0.1 d} > n^{1-\frac{2d}{3}}$}\\
    \end{align*}
    
    Similarly, for the other direction we get:
    
    \begin{align*}
        \Pr{M^*_t > 0.9 \log n} 
        & \geq 
        \qty(1-\exp{-n^{0.1}})\cdot\qty(1-n^{1-\frac{2d}{3}})
            && \text{by \cref{lem:max-min-group-values}}\\
        & \geq 
        1-\exp{-n^{0.1}} - n^{1-\frac{2d}{3}} + O(1/n)
            && \text{for $d > 3$}\\
        & \geq 
        1 - 2 \cdot n^{1-\frac{2d}{3}}
            && \text{since  $\exp{-n^{0.1}} < n^{1-\frac{2d}{3}}$}\\
    \end{align*}
    
    We follow a similar calculation to prove upper and lower bounds on the $\missingV$ at time $t$:
    
    \begin{align*}
        \Pr{\missingV_t > 0.9 \log n} 
        & \geq 
        \qty(1- 0.9 \log n \cdot \exp{-n^{0.1}})\cdot\qty(1-n^{1-\frac{2d}{3}}) 
            && \text{by \cref{cor:first-missing-group}}\\
        & \geq 
        1- 0.9 \log n \cdot \exp{-n^{0.1}} -n^{1-\frac{2d}{3}} + O(1/n)
            && \text{for $d > 4$}\\
        & \geq 
        1 - 2 \cdot n^{1-\frac{2d}{3}}
    \end{align*}
    
    and for the other direction:
    \begin{align*}
        \Pr{\missingV_t < 3 \log n} 
        & \geq 
        \qty(1- n^{-\frac{3\log n}{2}})\cdot\qty(1-n^{1-\frac{2d}{3}}) 
            && \text{by \cref{cor:first-missing-group}}\\
        & \geq 
        1- n^{-\frac{3\log n}{2}} -n^{1-\frac{2d}{3}} + O(1/n)
            && \text{for $d > 3$}\\
        & \geq 
        1 - 2 \cdot n^{1-\frac{2d}{3}}
    \end{align*}

    Finally,we prove the number of agents holding $\group = k$ is between $[L^*_k, U^*_k]$. Let $I(G_k, t)$ denote the number of agents having $\group = k$ at time $t$, then:
    
    \begin{align*}
        \Pr{I(G_k, t) < \frac{5 \cdot n}{2^{k+2}}} 
        & \geq 
        \qty(1- \exp{n^{0.1}/48})\cdot\qty(1-n^{1-\frac{2d}{3}})
            && \text{by~\cref{cor:independent-sg-fraction-bounds}}\\
        & \geq 
        1 - 2 \cdot n^{1-\frac{2d}{3}}
    \end{align*}
    
    and 
    
    \begin{align*}
        \Pr{I(G_k, t) > \frac{3 \cdot n}{2^{k+2}}} 
        & \geq 
        \qty(1- \exp{n^{0.1}/32})\cdot\qty(1-n^{1-\frac{2d}{3}})
            && \text{by~\cref{cor:independent-sg-fraction-bounds}}\\
        & \geq 
        1 - 2 \cdot n^{1 - 2d / 3}.
        \qedhere
    \end{align*}
\end{appendixproof}

% \todo{DD: Should the following statement be in the proof of the corollary?}
% Finally, by using a Union bound over all time stamps $t$ we get a polynomial time holding time for the $\group$ values their distribution. 

\begin{toappendix}
    Let us summarize \cref{thm:dynamic-groups} in the corollary.
    \begin{corollary}\label{cor:dynamic-groups}
        Let $t \geq 30 \ln n$. Then, for large values of $n$, with probability at least $1- 4/{n^{-20}}$, we have:
        \begin{itemize}
            \item The $\missingV$ will remain in $\qty[0.9 \log n , 3\log n] $.
            
            \item The number of agents who hold $\group = k$ for $1 \leq k \leq 0.9 \log n$, will remain in $\left[\frac{3 \cdot n}{2^{k+2}}, \frac{5 \cdot n}{2^{k+2}}\right]$.
        \end{itemize}    
    \end{corollary}
    \begin{proof}
         Let $d=30$ in \cref{thm:dynamic-groups}. 
        %  Then, with probability at least $1- 4\cdot n^{-20}$, $\missingV_t$ will remain in $\qty[0.9 \log n , 3\log n] $. The agents generate a new $\group$ value at every interactions. By taking a union bound on all $n^{19}$ interactions during the $n^{18}$ time the bounds will follow. Same argument applies to the number of agents holding each $\group$ value.
    \end{proof}
\end{toappendix}

\subsection{Group detection}\label{sub:group-detection}
% \todo{DD: use language of loosely-stabilizing in fixed pop size, not pop size changing}
% \todo{DD: make sure we distinguish between the actual value of $n$, and what an agent \emph{thinks} $n$ is, for instance when describing what $\signalArray$ values that agent checks.}
% Let M be the estimation value stored in the agents' memory, let n be the population size ...

In the previous section we show that the set of present $\group$ values among the population will quickly (in $O(\log n)$ time) enter a small interval of values ($[1, 8 \cdot \log n]$) consistent with the population size.
% and remain in that interval (with high probability for polynomial time~\cref{cor:holding-dynamic-groups}). %Assuming $M$ is the population size estimation value stored in the agents' memory, 
In this section, we will prove the following:

\begin{itemize}
    \item The agents \emph{agree} about the presence of group values that are in $[\log \ln n, 0.9 \log n]$ after $O(\log n)$ time WHP.

    \item For a non-existing group value $q$,
    each agent will have $\signalArray[q] = 0$ in $O(q + \log n)$ time WHP. 
    % in which $M$ is the population size estimation value stored in the agents' memory.
\end{itemize}

We designed~\cref{protocol:reset-count-down} such that each agent in the $i$'th $\group$ \emph{boosts} the associated signal value by setting $\signalArray[i] = B_i$ (recall $B_i = \Theta(i)$). We will show by having at least $L_i$ agents boosting $\signalArray[i]$,
the whole population learns about the existence of the $i$'th $\group$ in $O(\log n)$ time with high probability. 
Intuitively, although $\signalArray[i]$ starts lower than $\signalArray[j]$ for $i < j$, so potentially dies out more quickly, it is also boosted more often since more agents have group value $i$.
Concretely, with $L_i$ agents responsible to boost signal $i$, and for all indices $\log \ln n < i < 0.9 \log n$ in the $\signalArray$ of the agents, $\Pr{\signalArray[i] = 0} < \exp{-\frac{2^{B_i}}{n/L_i}}$.

% \todo{DD: use language of loosely-stabilizing in fixed pop size, not pop size changing}
% In the previous section we showed that as long as the population size remains constant, the set of present $\group$ values among the population will also remain \dd{approximately} the same. In this section we will show that how the information about the set of present $\group$ values are propagated through the population. 
% Additionally, we will prove that the agents are constantly reminded about the presence of the $\group$ values in $[\log \ln n, 1/2 \log n]$ \dd{(thus all agents \emph{agree} about which of these group values are present)}, as long as the population size remain unchanged. 

% We designed~\cref{protocol:reset-count-down} such that each agent in the $i$'th $\group$ \emph{boosts} the associated signal value by setting $\signalArray[i] = B_i$ \dd{(recall $B_i = \Theta(i)$)}. We will show by having at least $L_i$ agents boosting $\signalArray[i]$,
% the whole population learns about the existence of the $i$'th $\group$ in $O(\log n)$ time with high probability. 
% \dd{Intuitively, although $\signalArray[i]$ starts lower than $\signalArray[j]$ for $i < j$, so potentially dies out more quickly, it is also boosted more often since more agents have group value $i$.}
% Concretely, with $L_i$ agents responsible to boost signal $i$, and for all indices $\log \ln n < i < 1/2 \log n$ in the $\signalArray$ of the agents, $\Pr{\signalArray[i] = 0} < \exp{-\frac{2^{B_i}}{n/L_i}}$.

% \todoi{DD: add glue text explaining this lemma intuitively}
Intuitively, the next lemma shows that if the $\group$ values are distributed as in \cref{lem:group-distribution}, 
then the whole population will learn about all the present $\group$ values above $\log \ln n$ within $O(\log n)$ time. 
Note that $\Pr{u.\signalArray[i] = j}$ is the probability that the agent $u$ has value $j$ in the $i^{th}$ index of its $\signalArray$. 
The following lemma is a restatement from~\cite[Section 5.1]{alistarh17detection}.

% \todoi{ME: make new variables instead of using $L_i$ and $B_i$. }
\begin{lemmarep}\label{lem:signal-appearance}
    In the execution of~\cref{protocol:reset-count-down}, 
    suppose that for each group value $\log \ln n < i < 0.9 \log n$, at least $A_i$ agents hold $\group = i$. 
    % \todo{DD: replace phrase ``maximum value for signal $i$ with a more precise definition: reference line of pseudocode maybe, or the phrase ``the value an agent $u$ sets $u.\signalArray[i]$ when it sees an agent $v$ with $v.\group = i$''}
    % Let $r_i$ denote the maximum value for signal $i$. 
    For every agent $u \in \calA$ let $u.\signalArray[i] = r_i$ when $u.\group = i$.
    Assuming each agent has at least $r_{i}$ interactions, 
    then for a fixed agent $u$ and index $i$, $\Pr{u.\signalArray[i] = 0} \leq \qty(1-\frac{A_i}{n})^{2^{r_i -1}}$.
\end{lemmarep}
%%%%%%%%%%%%%%%%%%%%%%%%%%%%%%%%%%%%%%%%%%
%%%%%%%% Let's make it time dependent:

\begin{appendixproof}
    We are analyzing each $\signalArray[i]$ field independently since the agents update each index with $a, b \rto \max(a-1, b-1, 0)$. Note that for the analysis of this protocol, we don't need to know the exact counts of the agents who have value $j$ for $0 < j \leq r_i$ for each signal $i$; instead, we show that the agents' values remains above $0$ for all indices $\log \ln n < i < 0.9 \log n$ using a backward induction. 
    
    Fix agents $u_1$, $u_2$, and $i$, let $s_1 = u_1.\signalArray[i]$ and $s_2 = u_2.\signalArray[i]$ be the values associated to index $i$ of their $\signalArray$ before their interaction. Let $s'_1 = u_1.\signalArray[i]$ and $s'_2 = u_2.\signalArray[i]$ represent the same field after their interaction. Note that $s'_1 \geq j$ if and only if $\max \qty(s_1, s_2) > j$. 
    % \todo{DD: use qty command to make parentheses correct size}
    To show the base case of our induction, consider the interaction between agent $u_1$ and agent $u_2$, assuming both had at least $1$ interaction:
    
    \begin{align*}
        \Pr{s'_1 < r_i-1}
        & = 
        \Pr{s_1 < r_i} \cdot \Pr{s_2 < r_i}\\
        & \leq
        \qty(1-\frac{A_i}{n})^2
    \end{align*}

    Assuming all agents had at least $r_i-j$ interactions we can calculate $\Pr{s'_1 < j}$ recursively:
    \begin{align*}
        \Pr{s'_1 < j}
        & = 
        \Pr{s_1 < j+1} \cdot \Pr{s_2 < j+1}\\
        & \leq 
        \qty(1-\frac{A_i}{n} )^{2^{r_i-j}}.
    \end{align*}
    
    Specifically, for $j = 1$, conditioning on all agents having at least $r_i$ interactions, we have:
    \begin{align*}
        \Pr{s'_1 < 1}
        & = 
        \Pr{s_1 < 2}\cdot \Pr{s_2 < 2}\\
        & \leq 
        \qty(1-\frac{A_i}{n} )^{2^{r_i -1}}.
        \qedhere
    \end{align*}

\end{appendixproof}

To use the previous lemma, we need to make sure that the agents wait for sufficiently long time such that each agent has at least $r_i$ interactions. The next corollary uses \cref{lem:signal-appearance} to derive bounds for the entire protocol using bounds from \cref{lem:group-distribution} for the distribution of the $\group$ values.
Also, \cref{lem:signal-appearance} takes a union bound over \emph{all} agents and group values $i$,
and uses the concrete value $r_i = B_i = 3 \cdot i + 1$ used in our protocol.

% since only agents with group value $i$ can set $\signalArray[i] = B_i$,
% then $\Pr{u.\signalArray[i] = B_i} = 1/2^i$ by equation ~\eqref{eq:group-prob}.

\begin{corollaryrep}\label{cor:signal-appearance}
    % \todo{DD: replace ``maximum'' similarly to the TODO note for previous lemma}
    % Let $u.\signalArray[i] = B_i$ when $u.\group = i$.
    For all $i > 0$ and for every agent $u \in \calA$, 
    assuming $B_i = 3 \cdot i +1$ let 
    $u.\signalArray[i] = B_i$ if $u.\group = i$.
    % denote the maximum value for signal $i$. 
    Suppose that for each group value $\log \ln n < i < 0.9 \log n$, at least $L_i$ agents hold $\group = i$. 
    Let $\beta \geq 8$; 
    then after $\beta \log n$ time, we have:
    \[
    \Pr{(\exists u \in \calA)(\exists i \in \{\log \ln n, \ldots, 0.9 \log n\})\ u.\signalArray[i] = 0} \leq 2 \cdot n^{1-0.9 \beta}
    \]
\end{corollaryrep}
\begin{appendixproof}
    Let us assume at least $L_i$ agents have $\group = i$ (recall definition of $L_i$ in~\cref{lem:group-distribution}). Suppose all agents having $B_i$ interactions, by~\cref{lem:signal-appearance}, for a fixed agent $u$ and for all values of $\log \ln n< i< 1/2 \log n$, we have $\Pr{u.\signalArray[i] = 0} \leq \qty(1-\frac{L_i}{n} )^{2^{B_i-1}}$. By~\cref{lem:group-distribution}, we know $L_i \geq \frac{3/4 \cdot n} {2^i}$ with probability at least $1-2\cdot \exp{\frac{n^{0.1}}{32}}$, thus we can bound $\Pr{u.\signalArray[i] = 0}$ as follows:

    \begin{align*}
        \Pr{u.\signalArray[i] = 0} 
        & \leq 
        \qty(1-\frac{3/4}{2^i})^{2^{B_i-1}}\\
        & \leq 
        \qty(1-\frac{3/4 \cdot 2^{2i}}{2^{3i}})^{2^{3i}} \quad \text{set $B_i = 3\cdot i + 1$}\\
        & \leq 
        \exp{-3/4 \cdot 2^{2i}}\\
        & \leq 
        \exp{- 2^{2\log \ln n -1}} 
            \quad \text{since $ i> \log \ln n$}\\
        & = 
        \exp{-2^{\log \ln^2 n -1}} \\
        & = 
        \exp{-\frac{\ln^2 n}{2}}\\
        & = 
        n ^ {-\frac{\ln n}{2}}
    \end{align*}
    By the union bound over all agents $u \in \calA$ and all values of $i$, we have 
    \[
    \Pr{(\exists u \in \calA)(\exists i \in \{\log \ln n, \ldots, \leq 0.9 \log n \})\ u.\signalArray[i] = 0} \leq n^{2-\frac{\ln n}{2}}.
    \]

    Now, we set $\beta$ such that each agent has at least $B_i = 3 \cdot i +1$ interactions for all values of $i$ in $[\log \ln n, 0.9 \log n]$. Let $\beta \geq 8$, by \cref{cor:fixed-interval-interactions-bound}, after $\beta \log n$ time (equivalent to $\frac{\beta}{\ln 2} \ln n$) all agents have at least $\frac{0.2 \cdot \beta}{\ln 2} \ln n) \approx 0.4 \beta \log n $ interactions with probability at least $1-n^{1-0.9 \beta}$.
    
    % for any $d \geq 3$, during $d n \ln n $ interaction, all agents have at least $0.2d \ln n$ interactions with probability $1-n^{1-2d/3}$. Setting $d \geq 11$, all agents have at least $2.2 \ln n > 3/2 \log n$ after $11 \ln n \approx 8 \log n$ time.
    
    By the union bound, $\Pr{u.\signalArray[i] = 0} \leq n^{1-0.9 \beta} + n^{2-\frac{\ln n}{2}} \leq 2 \cdot n^{1-0.9 \beta}$ for sufficiently large $n$.
\end{appendixproof}

% \begin{corollary}\label{cor:stationary-reset-count-down}
% Fix agent $u$. Set $M_i = 4 \cdot i$ in the stationary analysis of~\cref{protocol:reset-count-down}, for all values of $\log \ln n< i< 1/2 \log n$, $\Pr{u.\signalArray[i] \geq 1} \geq 1-n^{-3/4 \ln n} $.
% \end{corollary}
% \begin{proof}
% By~\cref{lem:reset-count-down-stationary}, we have for all values of $\log \ln n< i< 1/2 \log n$, $\Pr{\signalArray[i] \geq 1} = 1- \left(1-\frac{L_i}{n} \right)^{2^{M_i}} $. Substituting $L_i$  with $\frac{3n}{2^{i+2}}$ (by~\cref{lem:independent-sg-fraction-bounds}) we have:

% \begin{align*}
%     \Pr{\signalArray[i] \geq 1} &= 1- \left(1-\frac{L_i}{n} \right)^{2^{M_i}}\\
%     &\geq 1- \left(1-\frac{3/4}{2^{i}} \right)^{2^{M_i}}\\
%     &\geq 1- \left(1-\frac{3/4 \cdot 2^{2i}}{2^{3i}} \right)^{2^{3i}} \quad \text{set $M_i = 3i$}\\ 
%     &\geq 1- exp(-3/4 \cdot 2^{2i})\\
%     &\geq 1- exp(-3/4 \cdot \ln^2 n) \quad \text{since $ i> \log \ln n $}\\
%     &\geq 1- n ^ {-3/4 \cdot \ln n}
% \end{align*}

% \end{proof}

In the following lemma, we show that
when there is no agent holding $\group = i$, then $\signalArray[i]$ will become zero in all agents ``quickly'' with arbitrary large probability. To be precise, with no agent boosting signal $i$, $\Pr{u.\signalArray[i] = 0} \geq 1-n^{-\alpha}$ within $\Theta(B_i + \alpha \ln n)$ time WHP in which $B_i$ is the maximum value for signal $i$.
The lemma is a restatement from~\cite[Lemma 3.3]{burman2021timeoptimal} and~\cite[Lemma 1]{alistarh17detection}.

\begin{lemma}\label{lem:signal-fading-time}
For every agent $u \in \calA$ let $u.\signalArray[i] = B_i$ when $u.\group = i$.
% Let $B_i$ be the maximum value the agents set for signal $i$. 
Assume that no agent sets its $\group$ to $i$ from this point on. 
Then for all $\alpha \geq 1$, all agents will have $\signalArray[i] = 0$ after $3n \ln (n^{\alpha} \cdot 3^{B_i})$ interactions with probability at least $1-n^{-\alpha}$.

\end{lemma}
\begin{proof}
Set $t = 3n \ln (n^{\alpha} \cdot 3^{B_i})$ and $R_{max} = B_i$ in the proof of \cite[Lemma 3.3]{burman2021timeoptimal}. 
\end{proof}
\subsection{Dynamic size counting protocol analysis}\label{subsec:dynamic-counting-protocol-analysis}

Let $\est$ denote the estimate of $\log n$ in agents' memory. 
Let $n$ be the true population size. 
In the previous section we show that the set of present $\group$ values among the population will quickly (in $O(\log n)$ time) enter a small interval of consecutive values ($[1, 3 \cdot \log n]$) consistent with the population size. In this section, we will show that the $\group$ values will remain in that interval (with high probability for polynomial time~\cref{thm:holding-time-of-protocol}). Moreover, the next two lemmas, show how the agents update their $\est$ in case it far from $\log n$.

Assuming the agents' $\est$ is much smaller than $\log n$, the next lemma shows that all the agents will notice the large gap between $\est$ and $\missingV$. Hence, they will re-calculate their population size estimate.

\begin{lemmarep}\label{lem:change-inc}
    Let $M = \max_{u \in \calA} u.\est$. 
    Assuming $M \leq 0.22 \log n$, then the whole population will enter $\wPhase$ in $O(\log n)$ time with probability at least $1-O(n^{-2})$.
\end{lemmarep}
\begin{appendixproof}
    
    \begin{enumerate}
        \item The whole population learns (i.e., $u.\est = M$ for all $u \in \calA$) about $M$ by epidemic in less than $3 \ln n$ time with probability at least $1-n^{-2}$~\cite[Corollary 2.8]{burman2021timeoptimal}. 
        
        \item By \cref{cor:signal-appearance}, for all $\group$ values in $[\log \ln n, 0.9 \log n]$, the agents will have a non-zero value in their $\signalArray$ after $10 \log n$ time with probability at least $1- 2 \cdot n^{-8}$ (setting $\beta = 10$). Thus, all agents have their $\missingV > 0.9 \log n$ by this time. 
        % 0.22 \log n< FMV< 3/4 \log n
        
        \item Since $M \leq 0.22 \log n$ by hypothesis, 
        and $\log n < \missingV / 0.9$ with probability $\geq 1- n^{-\frac{3\log n}{2}}$ by \cref{cor:first-missing-group},
        we have 
        $M \leq 0.22 \log n < 0.22 \cdot \missingV / 0.9 < 0.25 \cdot \missingV$.
    \end{enumerate}
    % Note that $0.9 \log n < \missingV < 3 \log n$ with probability $\geq 1- n^{-\frac{3\log n}{2}}$. Hence, $0.22 \log n < \missingV/4 < 3/4 \log n$. By \cref{protocol:check-size-change}, once an agent notices that $M < \missingV / 4 < 0.75 \log n$ they switch to $\wPhase$. 
    The lemma holds by the union bound over all three events.
\end{appendixproof}

% log n = 100
% M = 25
% FMV = 0.9 log n = 90
% agent looks for M < FMV / 4 = 22.5

For the other direction, assume  the population size estimate in agents' memory is much larger than $\log n$. We prove in the the following lemma that all the agents will notice the large gap between $\est$ and $\missingV$. Hence, they will re-calculate their population size estimation.

Note that in the \cref{cor:signal-appearance}, we proved for all $\group$ values $i$ for $\log \ln n \geq i$, the $\signalArray[i] $ will have a positive value in $O(\log n)$ time. However, we could not prove the same bound for values that are less than $\log \ln n$. So, inevitably the agents ignore their $\signalArray$ for values that are less than $\log \ln n$. Since the agents have no access to the value of $\log n$, they have to use $\est$ as an approximation of $\log n$. Thus, they ignore indices than are less than $\log M$ in $\signalArray$. Making the $\missingV$ a function of $\max(\log M , \log n)$.
For example, let us the true population size is $n$ but $M > 2^n$, then, the agents should ignore appearance of a zero in their $\signalArray$ for all indices $i$ that are $\leq \log(M) = n$. The correct $\missingV$ happens at index $j = \Theta(\log n)$ but the agents stay in the $\nPhase$ as long as $\signalArray[i]$ for $i \geq n$ are positive. Since for each $\signalArray[i]$ it takes $O(i)$ time to hit zero, it takes $O(n)$ time for the agents to switch to $\wPhase$.

Note that with our current detection scheme, for indices $i$ that are less than $\log \ln n $, the event of $\signalArray[i] = 0$ happens frequently. 

% \todo{DD: state some intuition about the two cases and why we cannot show it will converge in time $O(\log n)$ if $M > 2^n$}

\begin{lemmarep}\label{lem:change-dec}
    Let $M = \max_{u \in \calA} u.\est$. 
    Assuming $M \geq 7.5 \cdot \log n$, 
    then the whole population will enter $\wPhase$ in $ O(\log n + \log M)$ time with probability at least $1-O(n^{-2})$.
\end{lemmarep}
\begin{appendixproof}
    \begin{enumerate}
        \item The whole population learn about $M$ by epidemic in less than $3 \ln n$ time with probability at least $1-{n^{-2}}$~\cite[Corollary 2.8]{burman2021timeoptimal}. 
        
        \item By \cref{cor:signal-appearance}, for all $\group$ values in $[\log \ln n, 0.9 \log n]$, the agents will have a non-zero value in their $\signalArray$ after $10 \log n$ time with probability at least $1- 2 \cdot n^{-8}$ (setting $\beta = 10$).
        Moreover, any signal associated to a non-existing $\group$ value $q$, will be gone in $O(\alpha \ln n + q)$ time with probability $1-n^{1-\alpha}$ (\cref{lem:signal-fading-time}). 
        
        \item Additionally, by \cref{cor:first-missing-group}, the first missing $\group$ value $\missingV$ is less than $3 \log n$ with probability at least $1-n^{-\frac{3\log n}{2}}$. However, since \cref{cor:signal-appearance} only works for $\group$ values greater than $\log \ln n$, 
        an agent $u$ ignores any $u.\signalArray[i] = 0$ for $i < \log (u.\est) \leq \log M$.

        \item Since $M \geq 7.5 \log n$ by hypothesis, 
        and $\log n > \missingV / 3$ with probability $\geq 1- n^{-\frac{3\log n}{2}}$ by \cref{cor:first-missing-group},
        we have 
        $M \geq 7.5 \log n > 7.5 \cdot \missingV / 3 > 2.5 \cdot \missingV$.
        
    \end{enumerate}
    
    % the agents ignore any $\signalArray[i] = 0$ for $i < \log M$. 
    % So, in the extreme case that $\missingV < \log M$, the agents set their $\missingV = \log M$.
    
    % regular case
    % log M = 10
    % M = 1024 = 8 log n,
    % then log n = 128,
    % log log n = 7
    % FMV < 300
    % ignore values 10

    % extreme case
    % log M = 1000
    % M = 2^{1000} = ?,
\end{appendixproof}

% \todoi{DD: add glue text}
In the next theorem, we will show once there is large gap between the maximum $\est$ among the population and the true value of $\log n$, the agents update their estimate in $O(\log n + \log M)$ time.  

% \todoi{change statement to be consistent with protocol}

\begin{theoremrep}\label{thm:w-phase}
     Let $M = \max_{u \in \calA} u.\est$. 
     Assuming $\est \geq 7.5 \log n$ or $\est \leq 0.2 \log n$,
     %Let $\wPhase = 12 \cdot \max(\log M , \missingV)$ in~\cref{protocol:timer-routine}. 
     %Assuming every agent enters $\wPhase$, 
     then every agent replaces its $\est$ with a new value that is in $[\log n - \log \ln n, 2\log n]$ with probability $1-O(1/n)$ in $O(\log n + \log M)$ time.
\end{theoremrep}

% \begin{lemma}\label{lem:fixed-interval-interactions-bound}
%     For any $d \geq 3$ and $2d/3 \geq c \geq 1$ during $d n \ln n $ interaction, all agents have at least $2(d-\sqrt{dc}) \ln n$  and at most $2(d+\sqrt{3dc/2}) \ln n$ interactions with arbitrary large probability $1-n^{1-c}$. 
% \end{lemma}
\begin{proofsketch}
    By \cref{lem:change-dec,lem:change-inc}, once an agent notices the large gap between $\est$ and $\missingV$ they switch to $\wPhase$. We set $\wPhase$ long enough so when the first agent moves to $\cPhase$, there is no agent left in the $\nPhase$. Thus, they all re-generate a new geometric random variable and store the maximum as their $\est$.
\end{proofsketch}
\begin{appendixproof}
    Recall that by \cref{lem:change-dec,lem:change-inc}, once an agent notices the large gap between $\est$ and $\missingV$ they switch to $\wPhase$
    (i.e., if $\est > 2.5 \missingV$ or $\est < \missingV/4$). 
    
    By~\cref{cor:signal-appearance}, $10 \log n$ time is long enough so that all agents have $\signalArray[i] > 0$ for $\log \ln n<i < 0.9 \log n$ with probability $1-n^{-8}$ (setting $\beta=10$).
    
    By \cref{lem:fixed-interval-interactions-bound}, during $10 \log n \leq 15 \ln n$ time each agent has at most $60 \ln n \leq 44 \log n$ interactions with probability $\geq 1- n^{-3}$ (setting $d = 15$ and $c= 10$).

    Thus, each agent should count up to $28 \log n$ during $\wPhase$ to ensure that 
    to ensure that with high probability, all other agents have also entered $\wPhase$ before switching to $\cPhase$. By \cref{protocol:timer-routine}, once the $\wPhase$ is finished, an agent switches to $\cPhase$. 
    
    We set $t_w = 60 \cdot \missingV = O(\log n + \log M)$
    be the number of interactions each agent spends in phase $\wPhase$ in~\cref{protocol:timer-routine}.
    So, after $10 \log n$ time, every agent has $\missingV > 0.9 \log n$ (see \cref{protocol:update-missing-value}).
    Eventually, all agents switch from $\wPhase$ to $\cPhase$ to re-generate a new geometric random variable for $\est$. By the end of $\cPhase$ all agents have generated a new geometric random variable and propagate the maximum by epidemic in at most $3 \ln n$ time with probability at least $1-\frac{1}{n^2}$~\cite[Corollary 2.8]{burman2021timeoptimal}. By~\cite[Lemma D.7]{doty2018efficient}, the maximum of $n$ i.i.d. geometric random variables is in $[\log n - \log \ln n, 2\log n]$ with probability at least $1-\frac{1}{n}$.
    
    % \todo{DD: this sentence feels like it belongs outside of a proof, in an intuitive explanation of how the protocol works}
    % Since the agents have no access to the value of $\log n$, they use the $\missingV$ as an approximation of $\log n$. However, in our protocol such that the agents ignore their $\signalArray$ for values that are less than $\log M$. Thus, the length of $\wPhase$ and $\cPhase$ is a function of $\max(\log M , \missingV)$. 
    
    % Thus, we set $t_w = 60 \cdot \max(\log M , \missingV) \geq 44 \log n$ to achieve the above criteria. 
\end{appendixproof}

\begin{toappendix}
    %%%%% holding time theorem in the normal run of the protocol
    %%%%% space complexity lemma

    Finally, in the following lemma, we show that the holding time of our protocol is polynomial. 
    
    % low probability event of having agents going in the wphase. 
    
    \begin{lemmarep}\label{lem:prob-wphase-no-change}
        Consider the population after $45 \ln n$ time. Let $ 0.75 \log n \leq \est \leq 2.25 \log n$ for all agents. Then, an agent will remain in the $\nPhase$ with probability at least $1- 3 \cdot n^{-17}$.
    \end{lemmarep}
    
    \begin{appendixproof}
        An agent will change its phase from $\nPhase$ to $\wPhase$ if 
        $\est < \missingV/4 < 0.75 \log n$ 
        or 
        $\est > 2.5 \missingV > 2.25 \log n$ which will not happen as long as $\missingV \in [0.9 \log n, 3 \log n]$ with $ 0.75 \log n \leq \est \leq 2.25 \log n$. By \cref{cor:dynamic-groups}, $\missingV \in [0.9 \log n, 3\log n]$ with probability at least $1- n^{-\frac{3 \log n}{2}}$.
        
        However, there are two scenarios in which $u.\signalArray[i] = 0$ for some agent $u$ and $\log n \ln n < i < 0.9\log n$:
        
        1) There are less than $\frac{3 \cdot n}{2^{i+2}}$ agents setting their $\group$ equal to $i$; which happens with probability less than $4\cdot n^{-20}$ by \cref{cor:dynamic-groups}.
        
        2) The robust signaling protocol has failed to propagate signal $i$ which happens with arbitrary small probability $3 \cdot n^{-17}$ after $20 \log n (~ 14 \ln n)$ time (setting $\beta = 20$ in \cref{cor:signal-appearance}). 
        
        By the union bound over these two cases, the agents will remain in the $\nPhase$ with probability at least $1- 3 \cdot n^{-17}$
    \end{appendixproof}

    \begin{theorem}\label{thm:holding-time-of-protocol}
        Consider the population after $45 \ln n$ time. Let $ 0.75 \log n \leq \est \leq 2.25 \log n$ for all agents. Then, the agents will remain in the $\nPhase$ with probability at least $1- O(n^{-1})$ for $O(n^{15})$ time.
    \end{theorem}
    \begin{proof}
        By \cref{lem:prob-wphase-no-change}, the agents stay in $\nPhase$ with probability at least $1-3 \cdot n^{-17}$. By the union bound over all $n^{16}$ interactions in $n^{15}$ time, an agent might leave the $\nPhase$ with probability at least $1-O(1)/n$ during this time. 
    \end{proof}

    In the next lemma we calculate the space complexity of our protocol. Note that, the adversary can initialize agents with large integer values to increase the memory usage arbitrarily. 

    \begin{lemmarep}\label{lem:space}
        Assuming in the initial configuration of \cref{protocol:dynamic-counting-all} for every agent $u \in \calA$, we have $\max(u.\est, u.\RG, u.\group, u.\signalArray.size()) < s$. 
        % $O(S)$ states, 
        % After $O(M + \log n )$time ...
        % Assuming every agent has $ 0.75 \log n \leq \est \leq 2.25 \log n$, 
        then \cref{protocol:dynamic-counting-all} uses $O(\log^2 (s) + \log n \log \log n)$ 
        bits with probability at least $1-O(1/n)$.
    \end{lemmarep}
    \begin{appendixproof}
        The set of values that an agent stores are given through the different fields that are use by \cref{protocol:dynamic-counting-all}:
        
        \begin{align*}
            \est &\in [0.2\log n, 7.5\log n] = O(\log \log n)  \text{ bits}
                && \text{by \cref{thm:w-phase}} \\   
            \RG &\in [\log n - \log \ln n, 2\log n] = O(\log \log n)   \text{ bits}
                && \text{by \cite[Lemma D.7]{doty2018efficient}} \\   
            \group &= O(\log \log n)   \text{ bits}
                && \text{by \cref{lem:max-min-group-values}} \\   
            \timer &= O(\log \log n)  \text{ bits}
                && \text{by \cref{thm:w-phase}} \\   
            % \signalArray &=  O((\log n )!) = o(n^{\log \log n})   \text{states}
            %     && \text{by \cref{cor:max-min-group-values,cor:signal-appearance}}
            \signalArray &=  O(\log n \cdot \log \log n ) \text{ bits}
                && \text{by \cref{lem:max-min-group-values,cor:signal-appearance}}
        \end{align*}
    \end{appendixproof}

\end{toappendix}

\begin{theorem}
    Let $M = \max {u.\est}$ for all $u \in \calA$. 
    There is a uniform leaderless loosely-stabilizing population protocol that WHP:
    
    \begin{enumerate}
        \item If $M> 7.5 \log n$ or $M < 0.2 \log n$ reaches to a configuration with all agents 
            set their $\est$ with a value in $ [\log n - \log \ln n, 2 \log n]$ 
            in $O(\log n + \log M)$ parallel time.
            
        \item If $0.75 \log n < M < 2.25 \log n$, then the agents hold a stable $\est$ during the            following $O(n^{15})$ parallel time. 
        
        \item Assuming for every agent $u \in \calA$, $\max(u.\est, u.\RG, u.\group, u.\signalArray.size()) < s$ in the initial configuration, then the protocol uses $O(\log^2 (s) + \log n \log \log n)$ bits per agent. 
    \end{enumerate}
\end{theorem}

\subsection{Space optimization}\label{subsec:optimization}
In this section we explain how to reduce the space complexity of the protocol from $O(\log^2 (s) + \log n \log \log n)$ to $O(\log^2 (s) + (\log \log n)^2)$ bits per agent.

In \cref{protocol:dynamic-counting-all} the agents keep track of all the present $\group$ values using an array of size $O(\log n)$ (stored in $\signalArray$) by mapping every $\group =i$ to $\signalArray[i]$. 
We can reduce the space complexity of the protocol by reducing the $\signalArray$'s size. Let the agents map a $\group = i$ to $\signalArray[\floor{\log i}]$. So, instead of monitoring all $O(\log n)$ group values, they keep $O(\log \log n)$ indices in their $\signalArray$. Thus, reducing the space complexity to $O(\log^2 (s) + (\log \log n)^2)$ bits per agent.

Recall that in \cref{protocol:dynamic-counting-all}, there are $\approx \frac{n}{2^{i}}$ agents with $\group = i$ for $i\leq 0.9 \log n$ that help keeping $\signalArray[i]$ positive. 
However, with this technique, there will be $\approx \sum_{i = 2^j}^{2^{j+1}}\frac{n}{2^{i}}$ agents that are helping $\signalArray[i]$ to stay positive. So, every lemmas in \cref{sub:group-detection} about \cref{protocol:reset-count-down} hold. 

Finally, we update \cref{protocol:check-size-change} so that the agents compare their $\est$ with $2^{\logMissingV}$ in which $\logMissingV$ is the smallest index $i > \log \log M$ such that $\signalArray[i] = 0$.

On the negative side of this optimization, we get a protocol that is less sensitive about the gap between agents' \est\ and $\log n$.

\begin{theoremrep}\label{thm:space-optimization}
    Let $M = \max {u.\est}$ for all $u \in \calA$. 
    There is a uniform leaderless loosely-stabilizing population protocol that WHP:
    
    \begin{enumerate}
        \item If $M> 15 \log n$ or $M < 0.1 \log n$ reaches to a configuration with all agents 
            set their $\est$ with a value in $ [\log n - \log \ln n, 2 \log n]$ 
            in $O(\log n + \log M)$ parallel time.
            
        \item If $0.75 \log n < M < 2.17 \log n$, then the agents hold a stable $\est$ during the            following $O(n^{15})$ parallel time. 
    
        \item Assuming for every agent $u \in \calA$, $\max(u.\est, u.\RG, u.\group, u.\signalArray.size()) < s$ in the initial configuration, then the protocol uses $O(\log^2 (s) +(\log \log n)^2)$ bits per agent. 
    \end{enumerate}
\end{theoremrep}

\begin{appendixproof}
    Let $\logMissingV = \floor{\log(\missingV)}$, by \cref{cor:first-missing-group}, $0.9\log n < \missingV \leq 3\log n$ WHP. Equivalently, we can derive bounds for $\logMissingV$:
    
    \begin{align*}
        \log (0.9 \log n) & < \log (\missingV)  \leq \log (3\log n)\\
        \log \log n - 0.2 & \leq \log (\missingV) \leq \log \log n + \log 3\\
        \log \log n - 1.2 & \leq \floor{\log (\missingV)} \leq \log \log n + \log 3
    \end{align*}
    
    As the agents compare their $\est$ with $2^{\logMissingV}$, we need to calculate the bounds for $2^{\logMissingV}$:
    
    \begin{align*}
        \log \log n - 1.2 & \leq \floor{\log (\missingV)} \leq \log \log n + \log 3\\
        \frac{\log n}{2.3} & \leq 2^{\floor{\log (\missingV)}} = 2^{\logMissingV} \leq 3\log n
    \end{align*}
    which updates the bounds in \cref{lem:change-dec,lem:change-inc} and \cref{thm:w-phase} with the following:
    
    Having $M < 0.1 \log n \leq 0.1 \cdot 2.3 \cdot 2^\logMissingV \leq 0.23 \cdot 2^{\logMissingV}$.
    On the other side, having $M > 15 \log n \geq 15 \cdot 2^{\logMissingV}/3 \geq 5 \cdot 2^{\logMissingV}$.
    
    Finally, our optimization updates the bounds in and \cref{thm:holding-time-of-protocol} as follows:
    an agent will change its phase from $\nPhase$ to $\wPhase$ if 
    $\est < 2^{\logMissingV}/4 < 0.75 \log n$ 
    or 
    $\est > 5 \cdot 2^{\logMissingV} > 2.17 \log n$ which will not happen as long as $2^\logMissingV \in [\log n/(2.3), 3 \log n]$ with $ 0.75 \log n \leq \est \leq 2.17 \log n$.
    
    For the space complexity, note that bounds on $\signalArray$ in \cref{lem:space} changes to $O(\log \log n \cdot \log\log n)$ since there are about $O(\log \log n)$ indices with each index $i$ having a value that is $\Theta(i)$.
\end{appendixproof}

\section{Conclusion and open problems}
\label{sec:conclusion}

% \todo{DD: To save space, we don't need to re-summarize the results in the next two paragraphs. We can just go straight to open questions.}

In this paper, we introduced the dynamic size counting problem.
Assuming an adversary who can add or remove agents, the agents must update their $\est$ according to the changes in the population size. 
There are a number of open questions related to this problem.

% The agents should maintain an estimate of the population size proportional to the current value of the population size. We showed that by allowing the adversary to remove agents whenever it wants, solving the dynamic size counting problem is equivalent to solving the loosely-stabilizing counting problem. 

% We also provided a protocol that solves the loosely-stabilizing counting problem in time proportional to the current population size and the maximum value stored in agents' memory as the estimate. 
% Assuming $M$ is the maximum estimate value in the agents' memory, the \cref{protocol:dynamic-counting-all} converges after $O(\log n + \log M)$ time to a configuration in which all agents agree on an estimate that is $O(\log n)$ WHP. However, we do not know if it is possible to design a protocol with $O(\log n)$ convergence time. 

\noindent {\bf Reducing convergence time}.
Our protocol's convergence time depends on both the previous ($n_\mathrm{prev}$) and next ($n_\mathrm{next}$) population sizes, 
though exponentially less on the former:
$O(\log n_\mathrm{next} + \log \log n_\mathrm{prev})$.
Is there a protocol with optimal convergence time $O(\log n_\mathrm{next})$?

\noindent {\bf Increasing holding time}.
\cref{obs:impossible-to-stabilize} states only that the holding time must be finite, but it is likely that much longer holding times than $\Omega(n^{c})$ for constant $c$ are achievable.
For the loosely-stabilizing leader election problem, there is a provable tradeoff in the sense that the holding time is at most 
% \todo{DD: check this, I vaguely recall it's exponential but don't quite remember.}
exponential in the convergence time~\cite{izumi16,sudo2020}.
Does a similar tradeoff hold for the dynamic size counting problem?

\todo{DD: also state in terms of bits of memory}
\noindent {\bf Reducing space}.
Our main protocol uses $O(s + (\log n)^{\log n})$ states (equivalent to $O(\log^2 (s) + \log n \log\log n)$ bits). In \cref{subsec:optimization}, we showed how we can reduce the state complexity of our protocol to $o(n^{\epsilon})$ (equivalent to $O(\log^2 (s) + (\log \log n)^2)$ bits) by mapping more than one $\group$ to each index of the $\signalArray$ which reduces the size of the $\signalArray$ from $O(\log n)$ to $O(\log\log n)$. Another interesting trick is to replace our $O(\log n)$ detection scheme to $O(1)$ detection protocol of~\cite{dissemination} which puts a constant threshold on the values stored in each index. 
% \todo{DD: The next sentence makes it sound like this idea will provably work; it's not clear that it's just an idea and we don't know how to make it work. (The next paragraph slightly clears this up, but let's be more precise about the fact that this is an idea that doesn't work as stated due to the one-sided error.)}
So, it may be possible to reduce the space complexity even more to $O(c^{O(\log n)})$ (with all $O(\log n)$ indices present) or $O(c^{O(\log \log n)}) = \mathrm{polylog}(n)$ (using our optimization technique to have $O(\log \log n)$ indices in the $\signalArray$). 

However, the current protocol of ~\cite{dissemination} has a one-sided error that makes it hard to compose with our protocol. With probability $\epsilon > 0$, the agents might say signal $i$ has disappeared even though there exists agents with $\group = i$ in the population.

In the presence of a uniform self-stabilizing synchronization scheme, one could think of consecutive rounds of independent size computation. The agents update their output if the new computed population size drastically differs from the previously computed population size. 
Note that, the self-stabilizing clock must be independent of the population size since we are allowing the adversary to change the value of $\log n$ by adding or removing agents. To our knowledge, there is no such synchronization scheme available to population protocols. 

% \input{sections/sample}

%%
%% Bibliography
%%
\bibliography{my-library}

\end{document}